\documentclass[conference]{IEEEtran}
\IEEEoverridecommandlockouts
\usepackage{packages}

\DeclareRobustCommand*\cal{\relax\mathcal}

\newcommand{\D}{{\cal D}}
\newcommand{\X}{{\cal X}}

\newcommand{\K}{{\cal K}}
\newcommand{\Z}{{\cal Z}}

\newcommand{\E}{{\mathbb{E}}}
\newcommand{\N}{{\mathbb{N}}}

\newcommand{\R}{{\mathbb R}}

\newcommand{\adjac}{\approx}
\newcommand{\1}{\mathds 1}

\newcommand{\PR}[1]{\Pr \left( #1 \right)}
\newcommand{\conf}{\delta}
\newcommand*{\defsym}{\stackrel{\text{def}}{=}}

 \newcounter{ncomm}%

\newtheorem{theorem}{Theorem}

\newtheorem{lemma}{Lemma}
\newtheorem{remark}{Remark}
\theoremstyle{definition}
\newtheorem{definition}{Definition}[section]

\def\BibTeX{{\rm B\kern-.05em{\sc i\kern-.025em b}\kern-.08em
    T\kern-.1667em\lower.7ex\hbox{E}\kern-.125emX}}
\begin{document}
\sloppy
\title{On the (Im)Possibility of Estimating
Various Notions of Differential Privacy
}
\author{\IEEEauthorblockN{Daniele Gorla$^*$}
\IEEEauthorblockA{\textit{Dept. of Computer Science,}\\ \textit{Sapienza University of Rome (Italy)}}
\thanks{$^*$These two authors have equally contributed to the work.}
\and
\IEEEauthorblockN{Louis Jalouzot$^*$}
\IEEEauthorblockA{\textit{ENS de Lyon (France)}}
\and
\IEEEauthorblockN{Federica Granese}
\IEEEauthorblockA{\textit{Lix, Inria, Institute Polytechnique de Paris (France)} \\
\textit{ Sapienza University of Rome (Italy)}}
\and
\IEEEauthorblockN{Catuscia Palamidessi}
\IEEEauthorblockA{\textit{Lix, Inria, Institute Polytechnique de Paris (France)}}
\and
\IEEEauthorblockN{Pablo Piantanida}
\IEEEauthorblockA{\textit{L2S, CentraleSup\'elec (France)} \\
\textit{CNRS, Universit\'e Paris Saclay (France)}}
}

\maketitle

\begin{abstract}
We analyze to what extent final users can infer information about the level of protection of their data when the data obfuscation mechanism is a priori unknown to them (the so-called ``black-box" scenario). In particular, we delve into the investigation of two notions of \textit{local differential privacy} (LDP), namely $\varepsilon$-LDP and Rényi LDP. On one hand, we prove that, without any assumption on the underlying distributions, it is not possible to have an algorithm able to infer the level of data protection with provable guarantees; this result also holds for the central versions of the two notions of DP considered. On the other hand, we demonstrate that, under reasonable assumptions (namely, Lipschitzness of the involved densities on a closed interval), such guarantees exist and can be achieved by a simple histogram-based estimator. 
We validate our results experimentally and we note that, on a particularly well-behaved distribution (namely, the Laplace noise), our method gives even better results than expected, in the sense that in practice the number of samples needed to achieve the desired confidence is smaller than the theoretical bound, and the estimation of $\epsilon$ is more precise than predicted.
\end{abstract}


\section{Introduction}
Differential privacy (DP) \cite{Dwork06} is nowadays one of the best established and theoretically solid tools
to ensure data protection. Intuitively, given a set of databases,
differential privacy requires that databases that only slightly differ one from the other
(e.g. in one individual record) are mapped to the same obfuscated value with similar probabilities; this provides privacy
to the changed record because statistical functions run on the database should not overly  depend on the data of any individual.
The success of this privacy notion is testified by its wide application, both in academia
and in industry (see e.g. \cite{Apple,Dwork08,RAPPOR,4497436,5207644}). 

The classical notion of differential privacy relies on a central model where a trusted authority, the data collector, receives all the user data in clear without any privacy concerns.
It is then the responsibility of the data collector to privatize the received data by ensuring privacy and utility constraints.
More formally, the data collector applies to the collected data a randomized mechanism $\K$, that is a mapping from the set of databases to the distributions over $\Z$ (where $\Z$ is the domain of the obfuscated values). $\K$ is said to be $\epsilon$-DP (for some $\epsilon \geq 0$) if, for every $S \subset \Z$ measurable and every adjacent database $d_1$ and $d_2$ (i.e., databases that only minimally differ one from the other), we have that
    $$
        \Pr[\K(d_1) \in S] 
        \leq e^\epsilon \, \Pr[\K(d_2) \in S]
    $$
The value of $\epsilon$ controls the level of privacy: a smaller $\epsilon$ ensures a higher level of privacy.
Starting from this basic notion, a few variants have been proposed in the literature.

The first variant is a distributed version of DP, called {\em local differential privacy} (LDP) \cite{Alvim2018}.
In this setting, the users obfuscate their data by themselves before sending them to the data collector; therefore, the privacy constraints can be satisfied locally. 
Here, we do not work anymore with (adjacent) databases but directly on values from a set $\X$. Now, $\K$ is said to be $\epsilon$-LDP if, for every measurable $S \subset \Z$ and $x_1,x_2 \in \X$, it holds that
    $$
        \Pr[\K(x_1) \in S] 
        \leq e^\epsilon \, \Pr[\K(x_2) \in S]
    $$
Therefore, while in the central model it is the data collector that directly queries the user databases and applies the randomized mechanism, in the local model the idea is that the users apply the randomized mechanism to their own answers.
More abstractly, we could synthesize the two settings with a symbolic ``privacy wall": in DP, it is located after the data collector; in LDP, it is between the users and the data collector.
Again, $\epsilon$ controls the level of privacy: the smaller $\epsilon$, the higher the level of privacy. Furthermore, in the local setting, the users can choose the level of privacy they desire.

A second variant is the so-called {\em Rényi differential privacy} (RDP) \cite{RDP}, a relaxation of DP based on the notion of Rényi divergence. This notion has been introduced to remedy some unsatisfactory aspects of $(\epsilon,\delta)$-DP \cite{DKMMN06}, and it is well-suited for expressing guarantees of privacy-preserving algorithms and for composing heterogeneous mechanisms. 
More formally, a randomized mechanism $\mathcal{K}$ is $\epsilon$-RDP or $\epsilon$-LRDP of order $\alpha>1$ if, for every $x_1,x_2$ adjacent databases (for RDP) or values (for LRDP), we have that
    $$
        D_\alpha(\mathcal{K}(x_1) \| \mathcal{K}(x_2)) \leq \epsilon
    $$
where $D_\alpha(\ \cdot\ \|\ \cdot\ )$ denotes Rényi divergence \cite{Renyi}. Also here the value of $\epsilon$ controls the level of privacy, in the sense that a smaller $\epsilon$ corresponds to a higher privacy.

However, all such notions of DP are built into existing software products by the producing companies, and the final users have no way of testing the real level of security (i.e., the real value of $\epsilon$).
They can only trust the producers, sometimes leading to unexpected (and unwanted) behaviors.
For this reason, we would like to study to what extent final users can infer information about the level
of protection of their data when the data obfuscation mechanism is a priori unknown to them, and they can only sample from it (the so-called ``black-box" scenario). Indeed, in the literature, many examples exist \cite{AlbarghouthiH18,BartheGAHKS14,ChistikovMP19,NearDASGWSZSSS19,ReedP10,ZhangK17} (just to cite a very few) for testing the level of privacy when the obfuscation mechanism is known (the so-called ``white-box" scenario); by contrast,
to the best of our knowledge, almost no result is present for the black-box one.

\subsection{Our contribution}
In this paper, we focus on the local notions of DP and obtain the following results:
\begin{itemize}
    \item We prove that, without any assumption, no black-box algorithm exists for estimating the privacy parameter $\epsilon$ of LDP and LRDP. This impossibility result is quite strong, in the sense that it holds  no matter how low are the level of precision and confidence required. Furthermore, it holds for both continuous and discrete domains and noise distributions, and it can be adapted to the central versions of DP as well.
    \item Instead, if the densities involved are Lipschitz on a closed interval, we provide statistical guarantees showing the learnability of $\epsilon$ for both LDP and LRDP and for both continuous and finite discrete domains.
    \item Finally, we validate our method via experiments based on (a variation of) the Laplace obfuscation mechanism.
\end{itemize}
In particular, the paper is organized as follows.

In Section \ref{sec:prel} we provide the basic notions on differential privacy and $C$-Lipschitzness. 
We shall assume the randomized mechanism to be a probabilistic function $\K:\X\rightarrow \D\Z$, where 
$\D\Z$ denotes the set of probability distributions over $\Z$. We let $x$ range over $\X$, $z$ over $\Z$, $X$ denote a variable representing the value the randomized mechanism is applied on, and $Z=\K(X)$; then, for every $z \in \Z$, we have that $\K$ is $\epsilon$-LDP if and only if $p_{Z|X}(z|x_1) \leq e^\epsilon \, p_{Z|X}(z|x_2)$, where the notation $p_{Z|X}(\,\cdot\,|x_i)$ denotes either the density distribution of $\K(x_i)$ in the continuous case or the probability distribution \ $[z \mapsto \Pr(\K(x_i) = z)]$ \ in the discrete one.

In Section \ref{sec:impossibility} we first prove that, for every $\gamma > 0$ (precision), $0<\conf<1$ (confidence), and probabilistic estimator algorithm $\Tilde{\epsilon}$ that almost surely terminates, there exists a probability distribution $p_{Z|X}$ such that the estimated value for $\epsilon$ differs from the real one by at least $\gamma$ with probability at least $1-\conf$ (see \cref{thm:imposs}). Intuitively, no estimator can exist in the general case because we can always choose $p_{Z|X}$ to be non-regular and the points $z$ for which the real $\epsilon$ is reached can have a low probability of being sampled. Indeed, even though the estimator has access to the full range of values $\mathcal{Z}$, since it does not know the probability of each output, it cannot adapt (even on the fly) the number of samples from $p_{Z|X}$ it does to be sure to sample at least once each $z\in\mathcal{Z}$.

Then, in Section \ref{sec:est_ldp}, we focus on the continuous case and on the estimation of $\epsilon$ for a specific pair of values $(x_1,x_2)$. 
We start by showing that, whenever the densities $p_{Z|X}(\cdot|x_1)$ and $p_{Z|X}(\cdot|x_2)$ over $\mathcal{Z} = [a,b]$ are $C$-Lipschitz with $C< \frac 2 {(b-a)^2}$, there exists a probabilistic histogram-based estimator that succeeds and outputs an answer that differs from the real $\epsilon$ for at most $\gamma$ with probability at least $\conf$, for every $\gamma > 0$ (precision) and $0<\conf<1$ (confidence) (see \cref{thm:estimator}).
Once we have this estimator for a single pair of values, in Section \ref{subsec:actEst} we discuss how to estimate the overall $\epsilon$ (that is obtained as the sup of the $\epsilon$ for all the possible pairs). To this aim, we divide $\X$ in $k$ buckets (for a proper $k$), take the mid-points of all the buckets, run the previous estimator for all pairs of mid-points, and return the maximum. We prove that, if we also assume $p_{Z|X}(z|\cdot)$ to be $D$-Lipschitz, for some $D$ and for all $z \in \Z$, this new algorithm is able to estimate the overall $\epsilon$ with provable guarantees (see \cref{thm:buckets}).

Notice that the estimator allows for checking if an existing system actually provides privacy with budget $\epsilon$ only if the involved distributions are smooth (i.e., $C$-Lipschitz). Hence, a provider could lie both about $\epsilon$ and on the smoothness of the distributions; in the latter case, it is possible to trick the estimator into answering $\epsilon$ although the system does not actually provide such a level of privacy. Therefore, we have 2 cases to consider:
(1) the estimator agrees on the $\epsilon$ claimed by the provider, or (2) the estimator is in contradiction with the claims of the provider. In the second case, we can conclude that the provider should not be trusted because it lied either about $\epsilon$ or about the smoothness.
But in the first case, either the provider did not lie at all, or it lied about the smoothness of the distributions and tricked the estimator. To improve this unsatisfactory conclusion, we can implement a safety check derived from a necessary condition for Lipschitzness, that we provide in \cref{thm:scLip} in \cref{sec:Lipsch}. By making multiple executions of the estimator, one can empirically argue if the provider lied, by comparing the obtained probability with the theoretical lower bound from \cref{thm:scLip}.

In Section \ref{sec:exper} we first show that the Lipschitzness assumptions required by our theorems are met by the two most widely used DP mechanisms, namely Laplacian and Gaussian \cite{DworkMNS06,DworkR14}. Then, we validate all our results for the Laplace distribution. However, since our estimator assumes that $\Z$ is a closed interval, we have to rely on a variation of the Laplace distribution that we call {\em truncated Laplace distribution}. Since this distribution is not provided by classical libraries for sampling, we rely on the well-known inverse transform sampling to massively sample from the truncated Laplace.
We first consider the number of samples the estimator does; this parameter depends on $\gamma, \conf, C$ and $|\Z|$; however, we discover that the strongest dependency is on $\gamma$. So this parameter will have the largest influence on the time complexity. 
Then, we compare the estimated $\epsilon$ against the real one. We discover that the number of samples required to have satisfactory results in practice is significantly lower than the theoretical one. 
Furthermore, we study the proportion of estimated $\epsilon$ that are close to $\epsilon$ within $\gamma$ across 100 executions for different values of the number of samples. We discover that the lowest number of samples that yields a proportion greater than $\conf$ is around $400$ times lower than the theoretical one in this case.
Finally, we also test whether the safety check for $C$-Lipschitzness is practically valuable. The results show that it is, unless the actual $C$ is close to the claimed one. This is because \cref{thm:scLip} does not assume anything on the involved distributions, so its bound is certainly not tight for truncated Laplace distributions.

In Section \ref{sec:rdp} we adapt all these results to the notions of LRDP. The changes needed for the impossibility result are relatively minor, whereas the estimators, even if formally very similar, require more complex bounds both on the number of experiments and on the number of intervals required; furthermore, also the proofs of the desired guarantees are more technical. Then, we run experiments similar to the ones for LDP that confirm the quality of our approach also for LRDP. In particular, for this second setting the gap between the number of samples sufficient for achieving the guarantees of the theorem and the theoretical one (used in the proof of the theorem) is even more dramatic than for LDP: here the practical one is around 10$^5$ times smaller.

Finally, Section \ref{sec:concl} concludes the paper, by also drawing possible developments for future work.

All proofs are omitted or only sketched in this paper, but full details can be found in the appendix.

\subsection{Related work}
Many works in the literature face the problem of verifying DP of an algorithm, e.g. through formal verification. For instance, ShadowDP~\cite{WangDWKZ19} uses shadow execution to generate proofs
of differential privacy with very few programmer annotations and without relying on customized logics and verifiers; CheckDP~\cite{WangDKZ20} takes the source code of a mechanism along with its claimed level
of privacy and either generates a proof of correctness or a verifiable counterexample.
Finding violations of differential privacy has also been addressed by looking for the \textit{most powerful attack}~\cite{BichselSBV21} and by using \textit{counterexample generators}~\cite{BichselGDTV18,DingWWZK18}.
The above works assume that the program generating the obfuscation mechanism is known (white-box scenario), and uses the code to guide the verification process, 
with the exception of \cite{DingWWZK18}, which adopts a semi-black-box approach, 
in the sense that it is mostly based on analyzing the input-output of  executions, but also 
uses a  symbolic execution model to find values of  parameters in the code that makes
it easier to detect counterexamples. 

By contrast, very few approaches to verifying purportedly private algorithms in a fully black-box scenario (i.e., when the data obfuscation mechanism is unknown) have been proposed so far. One bunch of results are based on \textit{property testing}~\cite{GilbertM18,DixitJRT13}, where the goal is to design procedures to test whether the algorithm satisfies the privacy definition under consideration. Interestingly, given an oracle who has access to the probability density functions on the outputs, Dixit et al.~\cite{DixitJRT13} cast the problem of testing differential privacy on \textit{typical datasets} (i.e., datasets with sufficiently high probability mass under a fixed data generating distribution) as a problem of testing the Lipschitz condition. 
Their result concerns variants of differential privacy called \emph{probabilistic differential privacy} and \emph{approximate differential privacy} (i.e., $(\epsilon,\delta)$-DP), while we focus on the stronger notion of pure DP. Furthermore, unlike \cite{DixitJRT13}, we also provide an impossibility result, and we focus our analysis on LDP and RDP.
By contrast, Gilbert et al.~\cite{GilbertM18} prove both an impossibility result for DP and a possibility result for approximate DP. 
Their impossibility result shows that, for any $\epsilon > 0$ and proximity parameter $\alpha > 0$, no privacy property tester with finite query complexity exists for DP.
Their result does not cover ours, because their notion of "property testing algorithm" is a (probabilistic) decision procedure to prove or disprove DP, while we show that the impossibility holds even for a semi-decision procedure (disproving DP).
As for the possibility result, they provide lower bounds on the query complexity verification of approximate differential privacy, while, as already mentioned, we consider the stronger notion of pure DP.
Moreover, differently from our work, \cite{GilbertM18,DixitJRT13} achieve their possibility result using randomized algorithms.

The paper that is most closely related to ours is \cite{AKD22}.
Like us, the authors have also considered the problem of black-box estimating the $\epsilon$ of DP, but there are some differences. 
First, they only consider central DP, whereas we focus on LDP and LRDP, and then discuss how our results extend to central (R)DP. 
Second, we both consider a pair of databases/values and evaluate the $\epsilon$ for this pair; this provides a lower bound on the overall $\epsilon$, or proves that the provider lied on it. To compute a better under-approximation of the overall $\epsilon$, they iterate their method over a (somehow chosen) finite set of pairs; however, the strategy they propose has no provable guarantees. By contrast, our method, based on discretizing the set of values to be obfuscated and running the fixed-pair estimator on all pairs of mid-points, comes equipped with a formal statement that ensures guarantees on the estimation (see Subsection \ref{subsec:actEst}). 
Third, we use histograms and rely on the Lipschitzness of the noise function, whereas they use kernel density estimation and rely on Holder continuity (a generalization of Lipschitzness); however, we both assume that the noise is within a closed interval. 
Fourth, we provide the sample complexity, i.e., the number of samples (the $n$ defined implicitly in \eqref{ndef} for LDP and \eqref{nRDPdef} for LRDP) needed to achieve a certain precision $\gamma$ and confidence $\conf$ in the estimation of $\epsilon$ (Theorems \ref{thm:estimator} and \ref{thm:RDP-correct}). Instead, their theorem states that the estimation approximates $\epsilon$ asymptotically (within a certain confidence range) as $n$ grows, but they do not give the number of samples necessary to achieve the desired precision. 
Finally, we formally prove that it is impossible to estimate central and local (R)DP without any assumption; this is not provided in \cite{AKD22}.

Other two closely related papers are \cite{LO19}, and \cite{LWMZ22}, where the authors focus on $(\epsilon,\delta)$-DP and the estimation of the DP parameters of a given (unknown) mechanism. 
In \cite{LO19} the authors aim to estimate the parameters for a fixed pair of adjacent databases; their main focus is on the relation between the number of samples required and the accuracy of the estimation. To obtain the estimation of the parameters of the mechanism (i.e., by not relying on a specific pair of databases), they repeat their estimation on every possible pair.
In \cite{LWMZ22} the authors aim at estimating, once $\epsilon$ is given, the $\delta$ of a certain (unknown) mechanism. Like ours, their approach is black-box but is focused on  {\em polynomial-time} approximate estimators. They start by proving that an estimator (with provable guarantees) cannot exist; however, their result does not subsume ours, since we are not assuming the polynomiality of the estimator (hence, our impossibility results are more general than theirs) and because we also cover local notions of DP. To develop a poly-time estimator, they focus on the estimation of $\delta$ limited to a given subset $\cal T$ of all the possible databases, thus defining (and estimating) the notion of {\em relative DP}. This is done by estimating the $\delta$ over a given pair of adjacent databases in $\cal T$ (this is done using the kNN algorithm) and then taking the maximum over all possible such pairs. 

Finally, \cite{CCP19} and \cite{CCPT20} focus on estimating Bayes risk of a black-box mechanism and prove an impossibility result similar to ours. However, as it is noticed in the second paper (section VII), DP/LDP provide lower bounds for Bayes leakage (specifically, $\epsilon$-DP induces a bound on the multiplicative Bayes leakage, where the set of secrets are all the possible databases, and LDP implies a lower bound on Bayes leakage, but not vice versa). We want to remark that, being lower bounds, the impossibility proved for Bayes risk does not entail the impossibility results that we prove here.
\section{Preliminaries}
\label{sec:prel}

In this section, we recall some useful notions from differential privacy and on Lipschitz functions.

\subsection{$\epsilon$-Differential Privacy}
The classical notion of \textit{differential privacy} (DP) relies on a \textit{central} model where a trusted authority, the \textit{data collector}, receives all the user data in clear without any privacy concern.  It is then the responsibility of the data collector to privatize the received data by ensuring privacy and utility constraints.
More formally, given a set of databases $\X$,
DP relies on the notion of {\em adjacent} databases, whose most common definitions are: 
(1) two databases $x_1,x_2 \in \X$ are adjacent if $x_2$ can be obtained from $x_1$ by adding or removing one single record; 
(2) two databases $x_1,x_2 \in \X$ are adjacent if $x_2$ can be obtained from $x_1$ by modifying one record. 
We write $x_1 \adjac x_2$ to mean that $x_1$ and $x_2$ are  adjacent, under whichever definition of adjacency is relevant in the context of a given algorithm. The notion of adjacency used by an algorithm must be fixed a priori and it is part of the framework. 

A {\em randomized mechanism} for a certain query is any probabilistic function $\K:\X\rightarrow \D\Z$, where $\D\Z$ represents the set of probability distributions on $\Z$ and $z\sim\K(x)$ is the reported answer. 
Furthermore, we let $X$ denote a variable representing the database on which the randomized mechanism is applied and $Z=\K(X)$. 


\begin{definition}[$\epsilon$-DP \cite{DP}]
    \label{def:DP}\it
    A randomized mechanism $\K$ is said to be {\em $\epsilon$-differentially private} for some $\epsilon \geq 0$ if, for every $S \subset \Z$ measurable and every adjacent databases $x_1 \sim x_2 \in \X$, we have that
    $$
        \Pr[\K(x_1) \in S] 
        \leq e^\epsilon \Pr[\K(x_2) \in S]
    $$
    or equivalently if, for every $z \in \Z$, we have that
    \begin{align}
        p_{Z|X}(z|x_1) \leq e^\epsilon \, p_{Z|X}(z|x_2)
        \label{eq:DP}
    \end{align}
    where $p_{Z|X}(\, \cdot\,|x_i)$ denotes either the density distribution of $\K(x_i)$ in the continuous case or \ $[z \mapsto \Pr(\K(x_i)\! =\! z)]$ \ in the discrete one.
\end{definition}

\vspace{3pt}

    The value of $\epsilon$, called the {\em privacy budget}, controls the level of privacy: 
    a smaller $\epsilon$ ensures a higher level of privacy.
This parameter is what we aim at estimating throughout this paper.

\subsection{Local Differential Privacy}

Sometimes, it is more convenient to rely on 
a distributed variant of  $\epsilon$-DP, called \textit{local differential privacy} (LDP). In this setting, users obfuscate their data by themselves before sending them to the data collector~\cite{Alvim2018}; therefore, the privacy constraints can be satisfied locally. In this case, instead of working with a set of databases, the set $\X$ represents all possible values for the data of a generic user. 

\begin{definition}[Local differential privacy \cite{DuchiJW13}]\it
    \label{def:LDP}
    A randomized mechanism $\K : \X\rightarrow \D\Z$ is said to be {\em $\epsilon$-locally differentially private} for some $\epsilon \geq 0$ if, for every $z \in \Z$ and $x_1,x_2 \in \X$, it holds that
    \begin{align}
        p_{Z|X}(z|x_1) \leq e^\epsilon \, p_{Z|X}(z|x_2)
        \label{eq:LDP}
    \end{align}
    where again $p_{Z|X}(\,\cdot\,|x_i)$ denotes either the density distribution of $\K(x_i)$ in the continuous case or \ $[z \mapsto \Pr(\K(x_i) = z)]$ \ in the discrete one.
\end{definition}


\vspace{5pt}

Note that LPD implies DP on the collected data. In addition, LDP allows for a higher level of security, since the central authority can be considered malicious and men-in-the-middle attacks can be foreseen; indeed, LDP leaves to the users the freedom of choosing their own level of privacy. Anyway, the algorithms in the local
setting usually require more data than the algorithms in the central setting in order for the statistics to be significant~\cite{Alvim2018}. Some practical implementations of LDP are given by Google with \textsc{RAPPOR}~\cite{RAPPOR}, Apple~\cite{Apple}, and Microsoft \cite{MICROSOFT}.


\subsection{Rényi Differential Privacy}

The Rényi differential privacy \cite{RDP} is a relaxation of differential privacy based on the notion of Rényi divergence \cite{Renyi}, that we now recall.

\begin{definition}[Rényi divergence]\it
    For two probability distributions $P,Q$ with densities $p,q$ defined over the same space, the {\em Rényi divergence} of order $\alpha > 1$ is defined as
    $$
        D_\alpha(P\|Q)
        = \frac{1}{\alpha - 1} \log \E_{z \sim Q}
        \left(
            \frac{p(z)}{q(z)}
        \right)^\alpha.
    $$
\end{definition}

\vspace{5pt}

\begin{definition}[($\alpha,\epsilon)$-RDP~\cite{RDP}]\it
    \label{def:RDP}
    A randomized mechanism $\mathcal{K}$ is said to have {\em $\epsilon$-Rényi differential privacy} of order $\alpha>1$ if, for every $x_1,x_2$ adjacent databases for the central setting or values for the local one, we have that
    \begin{align}
    \label{eq:RDP}
        D_\alpha(\mathcal{K}(x_1) \| \mathcal{K}(x_2)) \leq \epsilon.
    \end{align}
\end{definition}

Notice that RDP can be extended by continuity to $\alpha \in \{1,+\infty\}$. 
In particular, it holds that
$\K$ is $(+\infty,\epsilon)$-RDP if and only if it is $\epsilon$-DP (the same holds for the local versions of RDP and DP, respectively).

\subsection{(Truncated) Laplace distribution}
\label{sec:truncated}
One of the main obfuscation mechanisms used to reach L(R)DP is replacing the value to obfuscate by a noisy one which is sampled from a Laplace distribution centered on the actual value \cite{DworkMNS06}. 

\begin{definition}\it
\label{def:Laplace}
A Laplace distribution $g$ of location $x$ and scale $B > 0$ is defined by the following density on $\R$:
$$
    g(z|x,B) = \frac{1}{2B}e^{- \frac{|z - x|} B}
$$
Its cumulative distribution function (CDF) is given by
$$
G(z)=\int_{-\infty}^z g(u|x,B)du
$$
and, for $p \in [0,1]$, its inverse CDF is
$$
G^{-1}(p)=x - B\operatorname{sgn}(p-0.5)\log(1 - 2|p-0.5|)
$$
\end{definition}

Notice that, to be interesting for DP, $\X$ (the set of possible values to obfuscate) must be a closed interval: indeed, a Laplace distribution is $\frac{\Delta \X}{B}$-DP, where $\Delta \X$ is the width of $\X$, as for all $z,x,x'\in\X$:
\begin{align*}
    \frac{g(z|x,B)}{g(z|x',B)}
    & = \exp\left(\frac{|z - x'| - |z - x|}{B}\right) \leq \exp(\Delta \X / B)
\end{align*}

However, our estimators will also require $\Z$ (the set of possible obfuscated values) to be a closed interval. So, for simplicity, in our experiments we will consider $\X = \Z$ and rely on an adapted version of Laplace mechanism that we call {\em truncated Laplace distribution}: this is a distribution defined over $\Z = [a, b]$ and, for a location $x \in \X$ and a scale $B > 0$, we consider the following density:
$$
    f(z|x,B)=K_{x,B}\ e^{-|z - x|/B}
$$
where the normalization constant $K_{x,B}$ is equal to:
$$
    \frac{1}{B(2 - e^{-(x - a)/B} - e^{-(b - x)/B})}
$$
Notice that both $f$ and $K$ also depend on the chosen interval $[a,b]$. However, to keep notation lighter, we will not specify this dependency in an explicit way.

\subsection{Lipschitz functions}

We recall here the definition of a classical property on functions that we will assume for densities later in the paper.

\begin{definition}[Lipschitz functions]\it
    \label{def:lipschitz}
    A function \ $f: \R \rightarrow \R$ \ is said to be {\em $C$-lipschitz}, for some constant $C \geq 0$, if, for every $x,y \in \R$, it satisfies
    $$
        |f(x) - f(y)| \leq C|x - y|
    $$
\end{definition}

    Notice that the smaller the constant $C$ the smoother the function $f$, and the less it varies.

To conclude this section, we now state a property of Lipschitz densities over a closed interval that will be useful in the proofs of our main results.

\begin{lemma}
    \label{lower}
    Let $p$ be a $C$-lipschitz density over $\mathcal{Z}=[a,b]$; then, for $W=b-a$ and for all $z\in\mathcal{Z}$, it holds that
    $$\frac{1}{W}-\frac{CW}{2} \leq p(z) \leq \frac{1}{W}+\frac{CW}{2}.$$
\end{lemma}

\section{On the estimation of $\epsilon$ in Local Differential Privacy}
\label{sec:DPLDP}
We start by examining whether it is possible to estimate the value of $\epsilon$ in the setting of LDP with fixed guarantees in a black-box scenario. By black-box we mean that the estimator does not know the 
internals of the mechanism or the distribution $p_{Z|X}(\cdot|x)$ generated by it, but can sample it as many times as needed. We assume that the estimator does this by calling an auxiliary procedure $\cal S(x)$ (the {\em sampler}) that samples in an i.i.d. way from $p_{Z|X}(\cdot|x)$. 
Furthermore, we wish the estimator to give sufficient guarantees, namely to provide a certain precision (the smaller the better) with a certain confidence (the higher the better).


First, we show that, without making any assumption about the distributions we want to estimate, such an estimator does not exist.
Then, we show that, under reasonable assumptions on the mechanism itself, we can design an estimator that can achieve \emph{any} level of precision with \emph{any} degree of confidence.

Notice that estimating \cref{eq:LDP} when \ $p_{Z|X}(z|x_2) > 0$ \ is equivalent to verifying
$\frac{p_{Z|X}(z|x_1)}{p_{Z|X}(z|x_2)}
    \leq e^\epsilon$.
Actually, when both $p_{Z|X}(z|x_1) > 0$ and $p_{Z|X}(z|x_2) > 0$, i.e., whenever $z \in \Z(x_1,x_2)$ where
    $$\Z(x_1,x_2) \defsym \{z\in \Z\ |\  p_{Z|X}(z|x_1)>0 \wedge p_{Z|X}(z|x_2)>0\},$$
we shall concentrate on the estimation of
$$
    \log\left(\frac{p_{Z|X}(z|x_1)}{p_{Z|X}(z|x_2)}\right)
    \leq \epsilon.
$$
We remark that all logarithms in this paper are natural.

For the negative result, we will show that it is impossible to estimate  $\epsilon$ even for a fixed pair of values $x_1$ and $x_2$. Namely, we  will consider the quantity
\begin{align}
\label{eq:epsx}
    \epsilon^{\star}(x_1, x_2)\defsym\underset{z\in\mathcal{Z}(x_1, x_2)}{\sup}\log\left(\frac{p_{Z|X}(z|x_1)}{p_{Z|X}(z|x_2)}\right).
\end{align}
%
In this setting, an estimator $\Tilde\epsilon$ is an algorithm that  takes in input a pair $(x_1,x_2)$, a precision $\gamma$  and a confidence $\delta$ (with $\gamma >0$ and $0<\delta<1$), and returns a real that is supposed to differ from $\epsilon^\star(x_1, x_2)$ for at most the precision with probability at least the confidence. Namely: 
$$ \Pr\left(\big|\epsilon^{\star}(x_1,x_2)-\Tilde{\epsilon}(x_1,x_2,\gamma,\delta)\big| \leq \gamma\right) \geq \conf.$$

For the positive result, of course we will need to range on all pairs, i.e., our final aim is to estimate  
\begin{align}
\label{eq:epsxx}
    \epsilon^{\star}(p_{Z|X})\defsym\underset{x_1, x_2\in\mathcal{X}}{\sup}\epsilon^{\star}(x_1, x_2).
\end{align}

\subsection{Impossibility result}
\label{sec:impossibility}

We now prove that \cref{eq:epsx} has no estimator with the desired guarantees; a fortiori, this proves that no estimator for \cref{eq:epsxx} exists.
We remark that this impossibility result is very strong: it shows that no estimator exists, even if (1) we are not very demanding about the precision and the confidence (namely, even if $\gamma$ is large and $\delta$ is small),
(2) even if the number of samples is unbounded and (3) the estimator is adaptive (namely, it can decide on the fly whether to stop or to continue sampling, based on previous samples).

\begin{theorem}[Impossibility]
    \label{thm:imposs}
    For every $\gamma > 0$ (precision), $0<\conf<1$ (confidence), 
    $x_1, x_2 \in \X$, and
    probabilistic estimator algorithm $\Tilde{\epsilon}$ that almost surely terminates, 
    there exists a probability distribution $p_{Z|X}$ such that 
    $$ \Pr\left(\big|\epsilon^{\star}(x_1,x_2)-\Tilde{\epsilon}(x_1,x_2,\gamma,\delta)\big| > \gamma\right) > 1-\conf.$$
\end{theorem}
\begin{proof}[Proof idea]
No estimator can exist in the general case because $p_{Z|X}$ has no regularity constraint and the points $z$ for which $\epsilon^{\star}$ is reached can have a low probability of being sampled. Even though the estimator has access to the full range of values $\mathcal{Z}$, since it does not know the probability of each output, it cannot adapt (even on the fly) the number of calls to $\cal S$ to sample at least once each $z\in\mathcal{Z}$.
See the appendix for full details.
\end{proof}

\begin{remark}
    The proof of~\cref{thm:imposs} assumes a binary setting both for $\mathcal{X}$ and $\mathcal{Z}$: an impossibility result in this setting clearly implies impossibility in more general ones, e.g., for continuous noise functions. Moreover, we do not use any property of $x_1$ and $x_2$. In particular, $x_1$ and $x_2$ could be the results of a query applied to two adjacent datasets; so, the result can be adapted also  to the  estimation of $\epsilon$ for central DP. 
\end{remark}

\subsection{Estimating $\epsilon^{\star}(x_1, x_2)$}
\label{sec:est_ldp}

Let us now impose a few reasonable
assumptions on the underlying probability distribution, namely that
$p_{Z|X}$ is defined over a closed interval $\mathcal{Z}\defsym [a, b]$, and that both $p_{Z|X}(\cdot|x_1)$ and $p_{Z|X}(\cdot|x_2)$ are $C$-Lipschitz, with $C < 2/W^2$ and $W=b-a$ (this condition on $C$ is needed to let $\tau$, as defined in \eqref{mdef}, be positive).
Once chosen the precision and confidence parameters, we will prove that $n$ samples from the two distributions (for a suitable value of $n$ that depends on $W$, $\gamma$, $\delta$ and $C$ -- see \cref{ndef} later on)
are enough for estimating \cref{eq:epsx} with guarantees. 


We propose a histogram-based estimator whose  pseudocode is 
provided in~\Cref{alg:histEstimator}.
For the desired precision $\gamma$, the estimator first divides $\mathcal{Z}$ into $m$ sub-intervals, each of width $w\defsym W/m$, where
    \begin{align}
        m \defsym \left\lceil\frac{6CW}{\tau\gamma}\right\rceil
        \label{mdef}
\qquad\text{and}\qquad
    \tau\defsym\frac{1}{W}-\frac{CW}{2}.
    \end{align}
In particular, 
we set $z_0 = a$ and $z_{j+1} = z_i + w$; 
one can readily check that $z_m = b$.
Then, the estimator chooses $n$ (the number of samples) such that
    \begin{align}
        2m(1-w\tau)^n+4f(n,w\tau,\gamma/12) \leq 1- \conf,
        \label{ndef}
    \end{align}
    where $f$ is defined as
    \begin{align}
    f(x,y,z) \defsym \frac{\exp\left(\frac{-xy(e^z-1)^2}{1+e^z}\right)+\exp\left(\frac{-xy(1-e^{-z})^2}{2}\right)}{1-(1-y)^x}
    \label{fdef}
    \end{align}
(note that $f$ is exponentially decreasing in $x$ and $y$). 
The estimator then invokes the sampler $n$ times both for $x_1$ and for $x_2$ (lines 4-7), counts the number of samples that appear in each sub-interval (lines 8-14), and considers these numbers as the approximations of $p_{Z|X}(\cdot|x_1)$ and $p_{Z|X}(\cdot|x_2)$ in that sub-interval; so, it computes their ratio and returns the highest value.

\begin{algorithm}[t]
    \caption{Histogram-based estimator for $\epsilon^{\star}(x_1, x_2)$}\label{alg:histEstimator}
    \begin{algorithmic}[1]
        \State \textbf{Input:} 
        $\mathcal{Z} (= [a,b]), \gamma, \delta, C$ 
        \State \textbf{Output:} $\Tilde{\epsilon}(x_1,x_2,{\cal Z},\gamma,\delta,C)$ \Comment{differing from $\epsilon^{\star}(x_1, x_2)$}
        \hspace*{4.7cm}for $\leq \gamma$ with prob. $\geq \conf$
        \State Compute $m$ and $n$ as in \cref{mdef} and \cref{ndef}, resp. \label{line0}
        \For{$1 \leq i \leq n$}
        \State $s_1[i] \gets \mathcal{S}(x_1)$
        \State $s_2[i] \gets \mathcal{S}(x_2)$
        \EndFor
        \For{$1 \leq j \leq m$}
            \State $N_j \gets \sum_{i}\1\left( z_j \leq s_1[i] < z_{j+1}\right )$ \label{lineA}
            \State $M_j \gets \sum_{i}\1\left( z_j \leq s_2[i] < z_{j+1}\right)$ \label{lineB}
            \If{$N_j = 0 \text{ or } M_j = 0$}
            \State fail
            \EndIf
        \EndFor
        \State \Return $\max_j\log\left(\frac{N_j}{M_j}\right)$  \label{lineR}
    \end{algorithmic}
\end{algorithm}

We are now able to state our theorem. 

\begin{figure}[t]
\framebox{
\centering
\begin{subfigure}[b]{0.45\textwidth}
\centering
\includegraphics[scale=7.5]{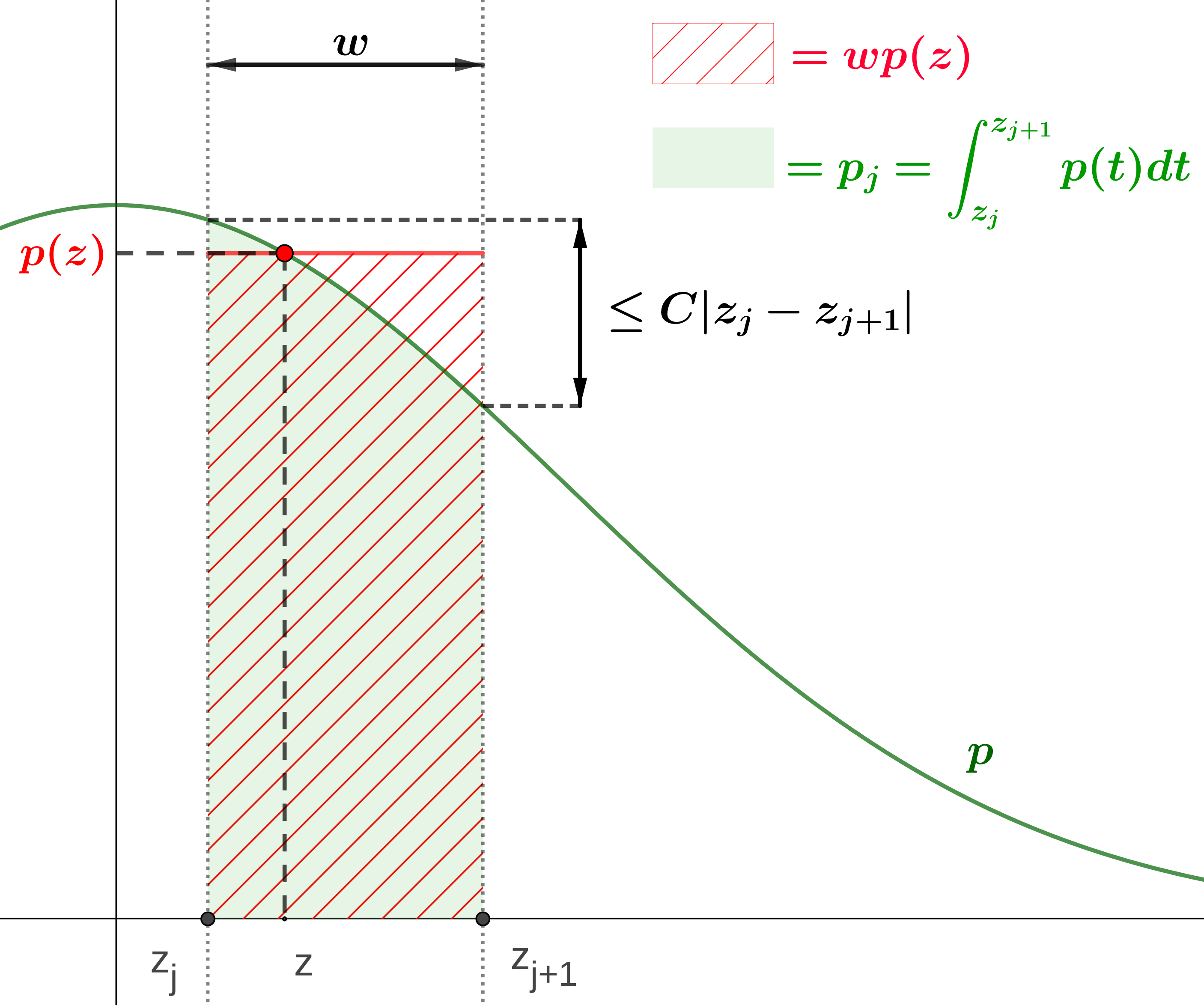}
\caption{}
\label{fig:smoothfig}
\end{subfigure}
}
\framebox{
\begin{subfigure}[b]{0.45\textwidth}
\centering
\includegraphics[scale=.35]{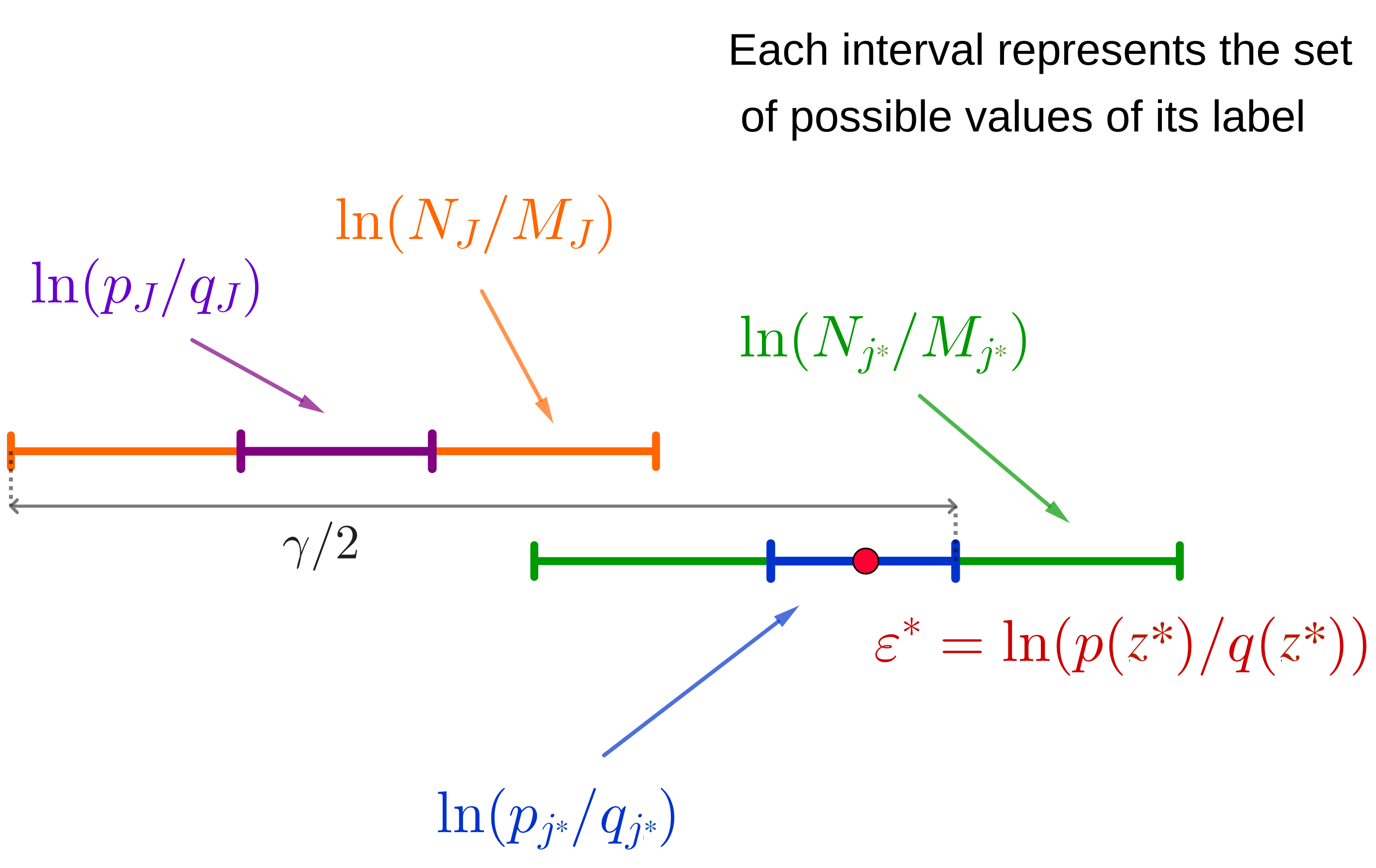}
    \caption{}
    \label{fig:summary}
\end{subfigure}
}
\caption{Proof sketch~\cref{thm:estimator}: (a) The smoothness of $p$ implies that $wp(z)$ is close to $p_j$. (b) Summary of the second and third steps of the proof, illustrating why the answer is close to $\epsilon^{\star}(x_1, x_2)$. The smallest intervals are bounded thanks to the smoothness of the distributions and the largest ones by the bound of \cref{chernoff}, implying that the empirical average tends to the expected value.
\vspace{-5mm}}
\end{figure}

\begin{theorem}[Correctness]
    \label{thm:estimator}
    Let densities $p=p_{Z|X}(\cdot|x_1)$ and $q=p_{Z|X}(\cdot|x_2)$ over $\mathcal{Z} = [a,b]$ be $C$-Lipschitz, with $C<2/W^2$ and $W=b-a$. For every $\gamma > 0$ (precision) and $0<\conf<1$ (confidence),
    we have that 
    \begin{align*}
       &\Pr\left(
       \begin{array}{l}
       \text{\Cref{alg:histEstimator} succeeds and }
        \\
         |\epsilon^{\star}(x_1, x_2) -\Tilde{\epsilon}(x_1,x_2,{\cal Z},\gamma,\delta,C)|\leq\gamma
         \end{array}
        \right )\geq \conf.
    \end{align*}
\end{theorem}
\begin{proof}[Proof sketch]
First, we know that, since $p$ is smooth, we can set the number of sub-intervals $m$ to reduce $w$ (the width of the sub-intervals) and have the probability $p_j$ of sampling within the sub-interval $j$ to be close to any $wp(z)$, with $z\in[z_j,z_{j+1})$ (see~\cref{fig:smoothfig} for an intuition). Being this true also for $q$, we get that \ $\log\left(\frac{p_{j}}{q_{j}}\right)$ \ is close to \ $\log\left(\frac{p(z)}{q(z)}\right)$.
In particular, this is true for some $z^{\star}$ for which $\epsilon^{\star}$ is reached.

Second, we consider the case where the binomial variables $N_j$ and $M_j$ are close to their average $np_j$ and $nq_j$. This happens with good probability and ensures that \ $\log\left(\frac{N_j}{M_j}\right)$ \ is close to \ $\log\left(\frac{p_{j}}{q_{j}}\right)$ and it is true thanks to the following Lemma:
\begin{lemma}
    \label[lemma]{chernoff}
    Let $X_1,\ldots,X_n$ be i.i.d.\! Bernoulli random variables with parameter $p>0$,
    and $S_n$ denote their sum (hence $\E(S_n) = np$). Then, for every $\gamma > 0$,
    $$\Pr(|\log S_n - \log np| > \gamma\ |\ S_n>0) \leq f(n,p,\gamma)$$
    where $f$ is defined in \eqref{fdef}.
\end{lemma}

Third, we have that the answer \ $\log\left(\frac{N_J}{M_J}\right)$ \ of the estimator is actually close to \ $\log\left(\frac{p_{j^{\star}}}{q_{j^{\star}}}\right)$, where $j^{\star}$ is the sub-interval containing $z^{\star}$, that is the value that $\epsilon^\star$ originates from (recall that, by definition of $\epsilon$-LDP, $\epsilon^\star = \sup_{z\in\mathcal{Z}}\log\left(p(z)/q(z)\right)$). Otherwise, if \ $\log\left(\frac{N_J}{M_J}\right)$ \ was much larger, we would have a contradiction with the definition of $\epsilon^{\star}$; similarly, if it was much smaller, it would imply \ $\log\left(\frac{N_J}{M_J}\right)<\log\left(\frac{N_{j^{\star}}}{M_{j^{\star}}}\right)$ \ and mean that the estimator had not returned the maximum value.

Finally, \ $\log\left(\frac{N_J}{M_J}\right)$ \ is close to \ $\log\left(\frac{p_{j^{\star}}}{q_{j^{\star}}}\right)$ \ which is close to \ $\log\left(\frac{p(z^{\star})}{q(z^{\star})}\right)=\epsilon^{\star}$; so the estimator works. See~\cref{fig:summary} for an illustration of this conclusion.
\end{proof}

To conclude, notice that, if
 $t$ denotes the time complexity of the sampler,
 the cost of sampling is $O(nt)$, whereas
the complexity of calculating each $N_j$ and $M_j$ is $O(n)$; therefore, the total complexity of Algorithm 1 is $O(n(t+m))$.

\begin{remark}
Our method is based on a discretization of the noise function. Obviously, it can be trivially adapted to the case in which the function is  discrete with finite domain provided a lower bound on the probability distribution. For instance, this is the case of the well known $k$-RR mechanism \cite{KOV16} (randomized response on a domain of $k$ elements).
\end{remark}

\subsection{On the $C$-Lipschitzness assumptions}
\label{sec:Lipsch}


Our estimator guarantees the bounds on the estimated $\epsilon$ only if the distributions are smooth; this is similar, e.g., to \cite{LWMZ22}. However, a provider could lie both about $\epsilon$ and about the smoothness of the distributions: it can trick the estimator into answering $\epsilon$ although the system does not actually provide such a level of privacy. Therefore, we have 2 cases to consider:
\paragraph*{Case 1} The estimator is in contradictions with the claims of the provider. Here we can conclude that the provider should not be trusted because it lied either about $\epsilon$ or about the smoothness.
\paragraph*{Case 2} The estimator agrees on the $\epsilon$ claimed by the provider. Then either the provider did not lie at all or it lied about the smoothness of the distributions and tricked the estimator.

Unfortunately, checking Lipschitzness of a given function is well-known to be a hard task \cite{JR11}. So, to improve the unsatisfactory conclusion drawn in the second case, we can implement a safety check derived from the following theorem:

\begin{theorem}[Necessary condition for Lipschitzness]
    \label{thm:scLip}
    Assuming that the distributions are $C$-Lipschitz, then, for any $c > 0$, the event
    \begin{align}
        \forall j. &
        \biggl[\frac{1}{n}|N_j-N_{j+1}| \leq 2c+Cw^2\ \text{ and }\nonumber\\
        & \ \, \frac{1}{n}|M_j-M_{j+1}| \leq 2c+Cw^2\biggr]
        \label{scLip}
    \end{align}
     happens with probability at least $1-8me^{-nc^2/3}$.
\end{theorem}

In practice, one can make multiple executions of \Cref{alg:histEstimator}, check if~\cref{scLip} occurs for each of them, and compute the empirical probability of this event (proportion of executions for which the event occurred).
If the empirical probability of~\cref{scLip} is significantly lower than the theoretical lower bound from~\cref{thm:scLip}, then it is likely (depending on the number of executions) that the distributions are not $C$-Lipschitz; so, the provider lied.



\begin{table*}[t]
    \centering
    \resizebox{1.5\columnwidth}{!}{
    \begin{tabular}{c|cccc||c|cccc}
    \toprule
    \multicolumn{5}{c}{Truncated Laplace}
    & \multicolumn{5}{c}{Truncated Gaussian}
    \\
    \cmidrule(lr){1-5} \cmidrule(lr){6-10}
    $B$ & $D$ & $C$ & $\epsilon$ & $C<2/W^2$ & $\sigma$ & $D$ & $C$ & $\epsilon$ & $C<2/W^2$ \\
    \midrule\midrule
    0.5 & 9.25 & 4.63 & 2.00 & \xmark &
    0.3 & 7.06 & 5.40 & 5.56 & \xmark
    \\
    0.8 & 4.38 & 2.19 & 1.25 & \xmark &
    0.5 & 2.42 & 2.03 & 2.00 & \xmark
    \\
    1 & 3.16 & 1.58 & 1.00 & \cmark &
    0.6 & 1.62 & 1.49 & 1.39 & \cmark
    \\
    2 & 1.27 & 0.64 & 0.50 & \cmark &
    1 & 0.54 & 0.70 & 0.50 & \cmark
    \\
    5 & 0.44 & 0.22 & 0.20 & \cmark &
    2 & 0.13 & 0.23 & 0.13 & \cmark 
    \\
    \bottomrule
    \end{tabular}
    }
    \caption{The Lipschitz constants ($D$ and $C$) for truncated Laplace and Gaussian mechanisms,  varying according to the value of their defining parameter (viz., $B$ and $\sigma$, resp.). We also display the resulting level of privacy $\epsilon$ and a check mark indicating if the guarantees of our theorems hold for the displayed values or not (recall that $\X\equiv\Z=[0,1]$ and so $W=1$).}
    \label{tab:DCconstants}
\end{table*}

\subsection{Estimating $\epsilon^{\star}(p_{Z|X})$}
\label{subsec:actEst}
As already said, \Cref{alg:histEstimator} estimates $\epsilon^{\star}(x_1, x_2)$, for a given pair $x_1,x_2$.
However, we can ask it to return the maximum between 
\ $\max_j\log\left(\frac{N_j}{M_j}\right)$ \
and
\ $-\min_j\log\left(\frac{N_j}{M_j}\right)=\max_j\log\left(\frac{M_j}{N_j}\right)$, thus 
avoiding switching $p$ and $q$ and call the estimator again.

A more delicate issue is that up to now we have only estimated the value of $\epsilon$ for a fixed pair $(x_1,x_2)$; so we need to generalize in this direction if we aim at estimating \cref{eq:epsxx}.
One solution is when $\mathcal{X}$ is finite: we can call the estimator for all 2-subsets $\{x_1,x_2\}$, resulting in 
$O(|\mathcal{X}|^2)$ calls; this is what is advocated, e.g., in \cite{LO19}. However, if $\X$ is big, the cost grows dramatically.
Alternatively, we can follow \cite{AKD22} and choose a subset of all possible pairs; this gives a statistical lower bound of the overall $\epsilon$. We can also follow \cite{LWMZ22}, where the authors focus on a {\em relative} notion of DP, i.e., DP restricted to a subset of $\cal X$; in this way, we estimate the DP parameter on this subset by running our estimator on all the pairs of this subset (that, indeed, is usually quite small). However, the latter two approaches reduce the computational cost at the price of losing guarantees on the estimation.

We now propose an alternative method that works well when $\X$ is infinite or very big; the price we pay is another mild assumption on the densities involved. Up to now, we have assumed that, for every $x \in \cal X$, $p_{Z|X}(\cdot|x)$ is $C$-Lipschitz, for $C < 2/W^2$ and $W$ is the width of $\cal Z = [a,b]$. Let us now also assume that, for every $z \in \cal Z$, $p_{Z|X}(z|\cdot)$ is $D$-Lipschitz, for some $D$; this new assumption is quite mild since we are not posing any constraint on $D$, and so any doubly differentiable function satisfies this requirement (since $\cal X$ is a closed interval). Then, the new algorithm that we propose first discretizes the interval $\cal X = [c,d]$ in $k$ buckets (for a proper value of $k$), then takes $x_i$ to be the mid-point of bucket $i$, runs \Cref{alg:histEstimator} for all possible $\{x_i,x_j\}$ and returns the maximum between all the estimated values, by ignoring the failed computations of $\Tilde\epsilon_{ij}$ when calculating the returned value.  
The details are given in \Cref{alg:histEstimatorAll}, whose overall time complexity is $O(k^2n(t+m))$; its correctness is provided by the following result, where we say that the algorithm \textit{succeeds} if at least one invocation of \Cref{alg:histEstimator} succeeds.

\begin{algorithm}[t]
    \caption{Estimator for $\epsilon^{\star}(p_{Z|X})$}\label{alg:histEstimatorAll}
    \begin{algorithmic}[1]
        \State \textbf{Input:} 
        $\mathcal{Z} (= [a,b]), \mathcal{X} (= [c,d]), \gamma, \delta, C, D$ 
        \State \textbf{Output:} $\Tilde{\epsilon}({\cal Z},{\cal X},\gamma,\delta,C,D)$ 
        \State Let $k \geq \frac{3D(d-c)}{\tau\gamma}$, where $\tau$ is defined in \cref{mdef}\label{kdef00}
        \State Divide $\cal X$ in $k$ buckets, with $x_i$ the mid-point of bucket $i$
        \ForAll{$\{x_i,x_j\} \subseteq \{1,\ldots,k\}$}
            \State 
            \label{lineTildeps}
            $\Tilde\epsilon_{ij} \gets \Tilde\epsilon(x_i,x_j,{\cal Z},\frac \gamma 3, \sqrt{\delta}\!, C)$, by invoking Alg.\ref{alg:histEstimator}
        \EndFor
        \State \Return $\max_{ij}\Tilde\epsilon_{ij}$
    \end{algorithmic}
\end{algorithm}

\begin{theorem}
\label{thm:buckets}
Let $\Z=[a,b]$, $\X = [c,d]$, and $p_{Z|X}$ be such that, for every $x \in \cal X$, $p_{Z|X}(\cdot|x)$ is $C$-Lipschitz, for $C < 2/W^2$ and $W=b-a$, and that, for every $z \in \cal Z$, $p_{Z|X}(z|\cdot)$ is $D$-Lipschitz, for some $D$. For every $\gamma > 0$ (precision) and $0<\conf<1$ (confidence),
    we have that 
    \begin{align*}
       &\Pr\left(
       \begin{array}{l}
       \text{\Cref{alg:histEstimatorAll} succeeds and }
        \\
        |\epsilon^{\star}(p_{Z|X}) -\Tilde{\epsilon}({\cal Z},{\cal X},\gamma,\delta,C,D)|\leq\gamma
        \end{array}
        \right )\geq \conf.
    \end{align*}
\end{theorem}


\subsection{Experiments}
\label{sec:exper}

We now test the validity of our approach by using real data obfuscation mechanisms.
However, since our estimators work only when $\X$ and $\Z$ are closed intervals, we rely on the {\em truncated} version of the mechanisms, as defined e.g. in \cref{sec:truncated}. For our experiments we set $\X=\Z=[0,1]$, so $W=1$. 

\bigskip
\noindent{\bf On the Lipschitness assumptions. }
We start by showing that the Lipschitzness requirements (that are needed to let our theorems hold) are met by two of the best known DP mechanisms, namely the (truncated) Laplace \cite{DworkMNS06} and Gaussian \cite{DworkR14} mechanisms. 
For the Gaussian mechanism, we consider the same process as for Laplace (see \cref{def:Laplace}) with 
$$g(z|x,\sigma) =
\frac 1 {\sigma\sqrt{2\pi}} e^{-\frac{(x-z)^2}{2\sigma^2}}
$$
Such functions are continuous with respect to $x$ or $z$ on closed intervals; so, they are Lipschitz and in \cref{tab:DCconstants} we depict the constants $D$ and $C$ that  they exhibit for various values of their defining parameter (viz., $B$ and $\sigma$). Furthermore, we also provide the resulting $\epsilon$ and we see in which cases the condition $C < \frac 2 {W^2}$ holds (and so the guarantees of our theorems are ensured).

From now on, we focus on the (truncated) Laplace mechanism, since it clearly enlightens the features of our approach.
To massively sample
from the CDF of a truncated
Laplace distribution, we follow the well-known inverse transform sampling method (details are provided in \cref{app_privacy:truncated}).

\bigskip
\noindent{\bf Number of samples. }
In~\cref{fig:nbSamples} we depict the logarithm of the number of samples $n$ in \Cref{alg:histEstimator} as functions of $\gamma$, $\conf$ and $C$. The analysis points out that $n$ strongly depends on $\gamma$. Since the time complexity of the estimator mostly depends on the number of samples (indeed, we always have that $m << n$), $\gamma$ is the parameter that most deeply influences the overall time complexity.

\begin{figure}[t]
\centering\framebox{
    \includegraphics[width=.47\textwidth]{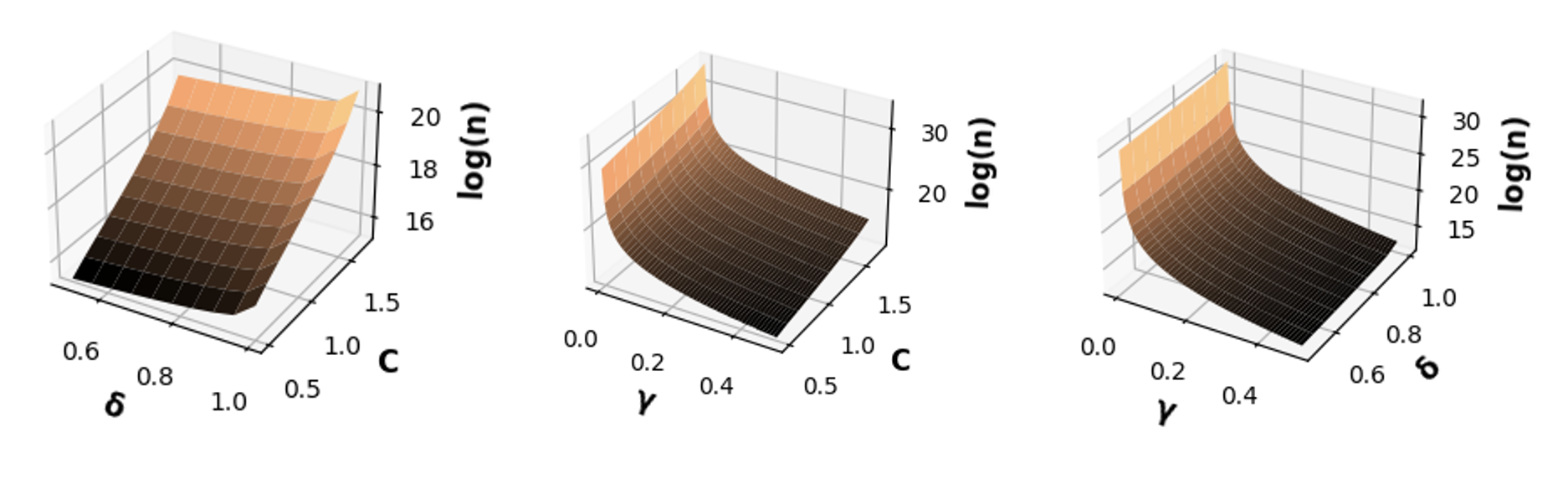}
    }
    \caption{Logarithm of the number of samples as functions of $\conf,\gamma$ and $C$ for LDP with $W=1$; when not varying, we set $C = 1$, $\conf = 0.9$ and $\gamma = 0.1$.}
    \label{fig:nbSamples}
\end{figure}



\bigskip
\noindent{\bf Analysis of \Cref{alg:histEstimatorAll}. }
In the following experiment, we 
set $\gamma=0.5$, $\delta=0.8$ and  $B=1$. The latter implies that $D\approx 3.16$, $C\approx 1.58$, and $\epsilon=1$; so, we set $n$ and $k$ (as defined in \Cref{alg:histEstimator} and \Cref{alg:histEstimatorAll}) to 1863132 and 91, respectively. In \cref{tab:est_all} we display $\Tilde{\epsilon}(x_i,x_j)$ for some pairs $x_i$ and $x_j$ among the $k^2 (=8291)$ computed during an execution of \Cref{alg:histEstimatorAll}. The maximum $\Tilde{\epsilon}(x_i,x_j)$ is 1, so the algorithm correctly returns 1. Notice that this value is reached only for $\{x_1,x_2\} = \{0,1\}$ (the extremes of $\X$). This result is not surprising since the Laplace density is unimodal, first increasing then decreasing; so, $\epsilon$ is reached when the modes are placed as far away from each other as possible. For this reason, we will focus our next experiments on the pair of extremes $\{0,1\}$, as any other pair would lead to an underestimation of $\epsilon$.

\begin{table}[t]
    \centering
    \resizebox{1\columnwidth}{!}{
    \begin{tabular}{c|c|c|||c|c|c}
    \toprule
    $x_i$ & $x_j$ &  $\Tilde{\epsilon}(x_i,x_j)$ &
    $x_i$ & $x_j$ &  $\Tilde{\epsilon}(x_i,x_j)$ \\
    \midrule
    0.00& 1.00& \textbf{1.00} & 0.54& 0.44& 0.12\\
    0.06& 0.49& 0.62 & 0.60& 0.94& 0.52\\
    0.11& 0.99& 0.96 & 0.67& 0.43& 0.26\\
    0.18& 0.48& 0.40 & 0.72& 0.93& 0.34\\
    0.23& 0.98& 0.89 & 0.79& 0.42& 0.45\\
    0.30& 0.47& 0.21 & 0.84& 0.92& 0.14\\
    0.36& 0.97& 0.80 & 0.91& 0.41& 0.65\\
    0.42& 0.46& 0.06 & 0.97& 0.91& 0.10\\
    0.48& 0.96& 0.66 & 1.00& 0.00& \textbf{1.00}\\
    \bottomrule
    \end{tabular}
    }
    \caption{$\Tilde{\epsilon}(x_i,x_j)$ for some pairs 
    among the $k^2$ 
    computed by \Cref{alg:histEstimatorAll} with a truncated Laplacian of parameter $B=1$, and with $\conf=0.8$ and $\gamma=0.5$.}
    \label{tab:est_all}
\end{table}

\begin{figure}[h]
    \centering
    \begin{subfigure}[b]{0.5\textwidth}
    \centering
    \includegraphics[width=\textwidth]{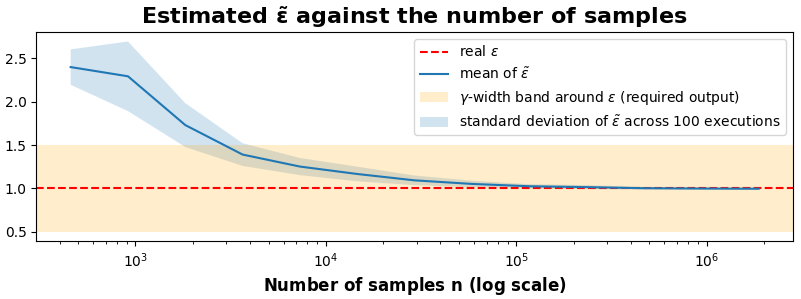}
    \caption{Analysis of $n$ on the precision $\gamma$.}
    \label{fig:eps_n}
    \end{subfigure}\\
    \begin{subfigure}[b]{0.5\textwidth}
    \centering
    \includegraphics[width=\textwidth]{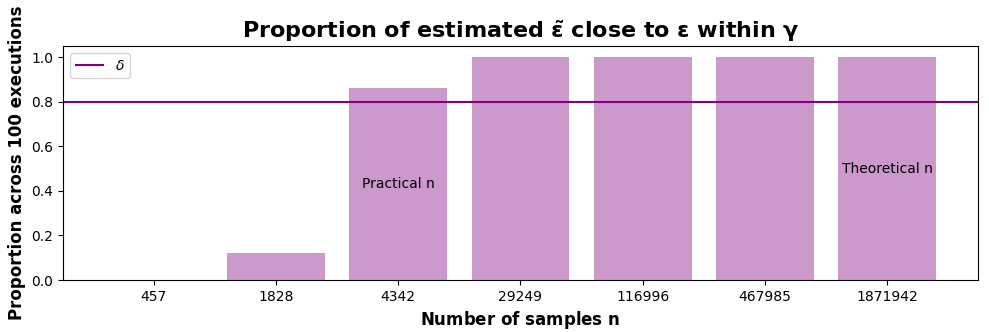}
    \caption{Analysis of $n$ on the confidence $\conf$.}
    \label{fig:proportion}
    \end{subfigure}
    \caption{Analysis on $n$ across 100 execution of~\cref{alg:histEstimator} for LDP.}
    \label{fig:my_label}
\end{figure}


\begin{table*}[t]
\centering
\resizebox{1.5\columnwidth}{!}{%
\begin{tabular}{r|ccc|ccc|ccc|ccc}
\toprule
\multirow{2}[3]{*}{$\conf=.8$} & \multicolumn{3}{c}
{\begin{tabular}{c}
     $\epsilon = .5$\\
     $C = .63$
\end{tabular}} & 
\multicolumn{3}{c}
{\begin{tabular}{c}
    $\epsilon = .7$\\
    $C = .97$
\end{tabular}} & 
\multicolumn{3}{c}
{\begin{tabular}{c}
    $\epsilon = 1$\\
    $C = 1.58$
\end{tabular}} & 
\multicolumn{3}{c}
{\begin{tabular}{c}
    $\epsilon = 2$\\ 
    $C = 4.62$
\end{tabular}} \\
\cmidrule(lr){2-4} \cmidrule(lr){5-7} \cmidrule(lr){8-10} \cmidrule(lr){11-13}
 & $n_{\text{TH}}$ & $n_{\text{PR}}$ & $m$ & $n_{\text{TH}}$ & $n_{\text{PR}}$ & $m$  & $n_{\text{TH}}$ & $n_{\text{PR}}$ & $m$  & $n_{\text{TH}}$ & $n_{\text{PR}}$ & $m$  \\
\midrule
$\gamma=1$   & 9588  & 56 & 6 &
               25488 & 112 & 12 &
               2.4e5 & 625 & 46 &
               Und. & 2600 & 100 \\
$\gamma=.5$  & 75618 & 262 & 12 &
               1.9e5 & 575 & 23 &
               1.9e6 & 4000 & 91 &
               Und. & 15600 & 200 \\
$\gamma=.1$  & 8.7e6 & 25600 & 56 &
               2.4e7 & 68000 & 114 &
               2.3e8 & 4.6e5 & 455 &
               Und. & 1.5e6 & 1000 \\
$\gamma=.05$ & 7e7   & 20800 & 112 &
               1.9e8 & 5e5 & 228 &
               1.9e9 & 3.3e6 & 909 &
               Und. & 1.2e7 & 2000\\
\bottomrule
\end{tabular}
}
\caption{Numbers of samples for different values of $\gamma$ and $C$ for a truncated Laplace mechanism for LDP. In every cell, we provide: the theoretical number of samples ($n_{\text{TH}}$), an order of magnitude for the practical one ($n_{\text{PR}}$), and the number of sub-intervals ($m$). When $C$ is too large, the theoretical number of samples is undefined and so is the number of sub-intervals, that has to be set by hand in practice (we used $m=\left\lceil100/\gamma\right\rceil$).}
\label{tab:theoPrac}
\end{table*}

\bigskip
\noindent{\bf Analysis of \Cref{alg:histEstimator}. }
Let us again consider the confidence  $\conf=0.8$, the precision $\gamma=0.5$, and a truncated Laplace of parameter $B=1$; then, the {\em theoretical} number of samples  $n$ (i.e., the one from~\cref{thm:estimator}) is $\approx 2 \times 10^6$. 
In~\cref{fig:eps_n} we show the mean and standard deviation of the estimated $\Tilde{\epsilon}$ across 100 executions of \Cref{alg:histEstimator} for different values of $n$. The real value of $\epsilon$ is shown with the red dashed line whilst the precision region is in yellow. As expected, the standard deviation decreases when the number of samples increases and all the  $\Tilde{\epsilon}$ estimated in every execution of the algorithm are close to $\epsilon$ within $\gamma$.
Notice that the number of samples required to have results that are satisfying in practice is significantly lower than the theoretical one (in~\cref{fig:eps_n}, the blue line and the red line overlap already when $n\approx 10^5$). This is not surprising since~\cref{thm:estimator} makes no assumption on the distributions, apart from Lipschitzness; so, the bound on the probability of success is far from being tight for the distributions considered in practice. 

In~\cref{fig:proportion}, we show the proportion of estimated $\Tilde{\epsilon}$ computed across 100 executions of~\Cref{alg:histEstimator} satisfying the precision requirement ($\gamma=0.5$). Indeed, we remind that such a proportion represents an empirical approximation for the probability in~\cref{thm:estimator}, which we want to be greater than $\conf$. With the theoretical number of samples, the empirical probability of estimating with the required precision $\epsilon$ is bigger than the requested confidence ($\conf=0.8$). From a practical perspective, in this simulation we are also interested in the lowest number of samples that satisfy the precision/confidence requirements. It turns out that the \textit{practical number of samples} is approximately $400$ times lower than the theoretical one in this case. In~\cref{tab:theoPrac}, we deepen the analysis by fixing $\conf$ and varying $\gamma$. 
To conclude, notice that, even if the assumptions of \cref{thm:estimator} are not met, the estimator still works, as we can see by looking at the last column of \cref{tab:theoPrac}.
However, one drawback when using the estimator without the theoretical backup is that the user has to set both $m$ (the number of sub-intervals) and $n$ (the number of samples), and needs to balance between them.

\bigskip
\noindent{\bf The safety check. }
We now aim to verify whether the safety check (w.r.t. $(x_1, x_2)$) in~\cref{thm:scLip} works in practice when the provider lies about the smoothness of the distribution. In particular, we suppose that the provider can lie about the $C$-Lipschitz parameter.

\begin{figure}
    \centering
    \includegraphics[width=0.5\textwidth]{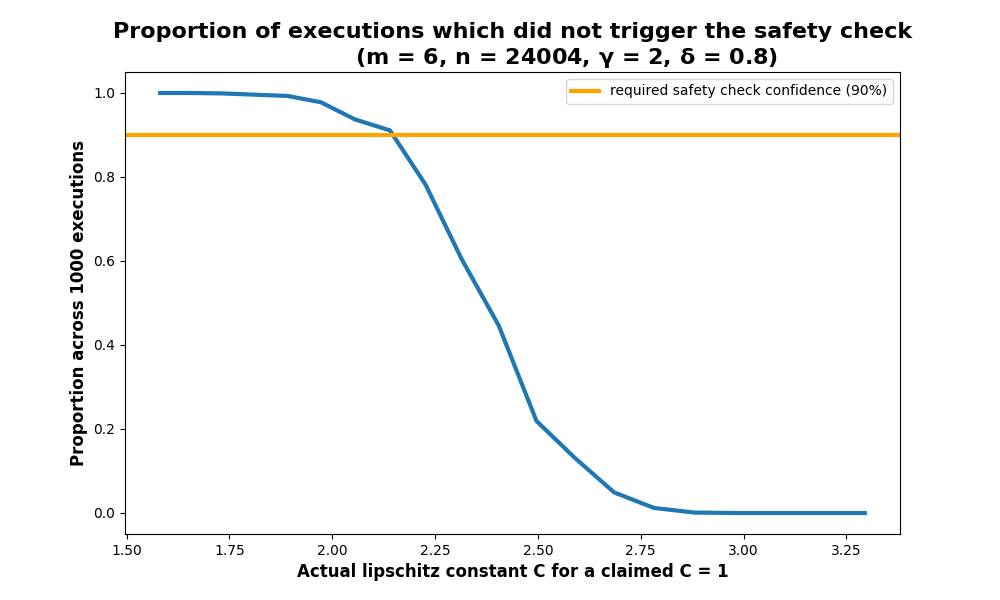}
    \caption{Experiments on the safety check.}
    \label{fig:safety}
\end{figure}

As in the section before, we start by fixing precision and confidence ($\gamma=2$ and $\conf=0.8)$, and we set the $c$ of \cref{thm:scLip} to $(Cw^2)/2$. Moreover, we choose some \textit{required probability} (say 0.9) and, if needed, we increase the number of samples $n$ so that the theoretical lower bound $1-8me^{-nc^2/3}$ on the probability of the safety check to be met is greater than our required probability. Then, we perform 1000 execution of~\Cref{alg:histEstimator} by considering first the $C$ claimed by the provider (i.e., $C=1$) and then the actual $C$ (i.e., the true one). Therefore, we look at the proportion of executions that do not trigger the safety check. This empirically approximates the probability of the safety check being met in~\cref{thm:scLip}.
The results are shown in~\cref{fig:safety}, where we can see that, if the actual $C$ is close to the claimed one (viz., when the curve is between 1.50 and 2), the safety check does not allow us to detect the lie from the provider. This is because, once again,~\cref{thm:scLip} does not assume the particular kind of distributions used, so its bound cannot be tight for truncated Laplace distributions.

\section{On the estimation of $\epsilon$ in Local R\'enyi Differential Privacy (LRDP)}
\label{sec:rdp}

We now adapt all previous results to LRDP.
As before, 
we first concentrate on a fixed pair $x_1, x_2\in\mathcal{X}$ to estimate 
\begin{align}
\label{eq:epsRDPx}
    \epsilon^\star_{\alpha}(x_1, x_2) \defsym 
    D_\alpha\left(p_{Z|X}(\cdot | x_1)\big\| p_{Z|X}(\cdot | x_2)\right).
\end{align}
Again, without reasonable assumptions it is  impossible to estimate~\cref{{eq:epsRDPx}} with guarantees. On the other hand, by having the same assumptions as for LDP, this quantity can be estimated as well.
Then, these results will transfer to the estimation of 
\begin{align}
\label{eq:sup_rdp}
    \epsilon^\star_{\alpha}(p_{Z|X}) \defsym \underset{x_1, x_2\in\mathcal{X}}{\sup}\epsilon^\star_{\alpha}(x_1, x_2).
\end{align}

\subsection{Impossibility result}

The impossibility result in~\cref{sec:impossibility} also holds for LRDP. This is more surprising than the analogous one for LDP: indeed, if there is again some output where the probabilities differ significantly but the probability of this output is low, then one would think that this would not violate the RDP guarantee since Rényi divergence averages over all outputs, instead of taking the pointwise maximum. However, as the following theorem proves, this is not the case.

\begin{theorem}[Impossibility]
\label{thm:RenyiImposs}
    For every $\alpha > 1$ (LRDP order), 
    $\gamma > 0$ (precision), 
    $0 < \conf < 1$ (confidence),
    $x_1, x_2 \in \X$, and
    probabilistic estimator  $\Tilde{\epsilon}_{\alpha}$ that almost surely terminates, there exists a probability distribution $p_{Z|X}$ such that
    $$ \Pr
        \left(
            \big| \epsilon_{\alpha}^{*}(x_1,x_2)
            - \Tilde{\epsilon}_{\alpha}(x_1, x_2, \gamma, \conf) \big|
            > \gamma
        \right) > 1 - \conf.$$
\end{theorem}
\begin{proof}[Proof Sketch]
Let us set the distributions
$$ p_{Z|X}(\cdot|x_1)= \mathcal{B}(d^{\frac 1 \alpha})
\ \text{ and } 
p_{Z|X}(\cdot|x_2)= \mathcal{B}
    \left(\frac{d^{\frac 1 {\alpha - 1}}} h
    \right)
$$
where $\mathcal{B}(p)$ denotes a Bernoulli distribution with parameter $p$, and we can consider the case when $\Tilde{\epsilon}$ samples only $0$s. In that case, the behaviour of $\Tilde{\epsilon}$ is independent of $d^{\frac 1 \alpha}$ and $\frac{d^{\frac 1 {\alpha - 1}}} h$ since its a priori distributions for $p_{Z|X}(\cdot|x_1)$ and $p_{Z|X}(\cdot|x_2)$ are $\mathcal{B}(0)$. Therefore, we can freely set $d$ so that the probability for $\Tilde{\epsilon}$ of sampling only $0$s is high enough and $h$ is sufficiently large, so that $\epsilon^{\star}\geq \log(h)$ is far enough from the answers of $\Tilde{\epsilon}$.
See~\cref{fig:bernoulliRDP} for an illustration of such Bernoulli distributions and the appendix for full details.
\end{proof}

\begin{figure}[!tbp]
 \framebox{\!\!\!
    \begin{tikzpicture}[thick]
        \def\red{red!70!black}
        \def\green{green!60!black}
        \node at (0, 0) (root) {};
        \draw (root) --+ (4.5, 0);
        \node[below] at (1.25, 0) (0) {$0$};
        \node[below] at (3.25, 0) (1) {$1$};
        
        \node[above left=3cm and .5cm of 0, left, color=\green] (oppd) {$1-d^{\frac 1 \alpha}$};
        \draw[color=\green] (oppd.east)+(.25, 0) --+ (1, 0);
        \node[above left=3.5cm and .5cm of 0, left, color=\red] (oppdh) {$1-\frac{d^{\frac 1 {\alpha - 1}}} h$};
        \draw[color=\red] (oppdh.east)+(.25, 0) --+ (1, 0);
        
        \node[above left=1cm and .5cm of 1, left, color=\green] (d) {$d^{\frac 1 \alpha}$};
        \draw[color=\green] (d.east)+(.25, 0) --+ (1, 0);
        \node[above left=.35cm and .5cm of 1, left, color=\red] (dh) {$\frac{d^{\frac 1 {\alpha - 1}}} h$};
        \draw[color=\red] (dh.east)+(.25, 0) --+ (1, 0);
        
        \node[color=\green, above right=3.3cm and .5cm of 1] (p) {$p_{Z|X}(\cdot|x_1)=
        \mathcal{B}(d^{\frac 1 \alpha})
        $};
        \node[color=\red, below=-.1cm of p] {$p_{Z|X}(\cdot|x_2)= \mathcal{B}
    \left(\frac{d^{\frac 1 {\alpha - 1}}} h
    \right)
        $};
        
        \draw[<->] (dh.east)+(1.4, 0) --+ (1.4, .5) node[right=.2cm, midway, text width=3.1cm] {\small{far so that $\epsilon^{\star}\geq \log(h)$\\ is large}};
        
        \draw[<->] (d.west)+(-.2,.05) --+ (-.2, 1.95) node[right=.2cm, midway, text width=3cm] {\small{far so that only\\ sampling 0s is likely}};
    \end{tikzpicture}\!\!\!
}
\caption{Illustration of the two Bernoulli distributions in the proof of Thm.\,\ref{thm:RenyiImposs}}
\label{fig:bernoulliRDP}
\end{figure}
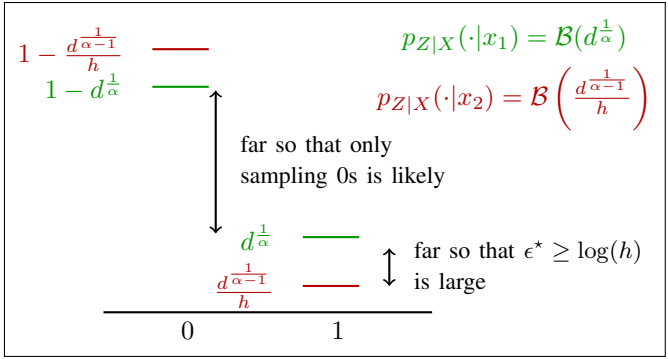

\medskip

We remark that also for RDP the previous result (for the local setting) holds for the central setting, where of course $x_1$ and $x_2$ are adjacent databases.

\subsection{Estimating $\epsilon_{\alpha}^{\star}(x_1, x_2)$}
For estimating LRDP, we follow the same setting as the one followed for LDP.
So, we still assume the distributions to be $C$-Lipschitz over $\mathcal{Z}\defsym [a, b]$, with  $a, b\in\mathbb{R}$, $C < 2/W^2$ and $W=b-a$

The new estimator, that we call Algorithm 3, can be adapted from \cref{alg:histEstimator}
by changing the values of $m$ and $n$, and the returned value. In particular, let us define
\begin{align}
    \tau_0 \defsym \frac{1}{W}-\frac{CW}{2}
    & \ \text{ and } \
    \tau_1 \defsym \frac{1}{W}+\frac{CW}{2},
    \\
    K \defsym \frac{2\tau_1^\alpha}{\tau_0^{\alpha-1}}
    & \ \text{ and } \
    K' \defsym \frac{\tau_0^\alpha}{\tau_1^{\alpha - 1}},
\\
\label{defGammaPr}
    \gamma' =
    & \min
        \left(
            \frac{\gamma K'(\alpha - 1)}{2K(2\alpha - 1)},
            \frac{\log 2}{2\alpha - 1}
        \right).
\end{align}
The new algorithm will divide $\mathcal{Z}$ in $m$ sub-intervals s.t.
\begin{align}
\label{defM}
    \frac{CwK(2\alpha - 1)}{2\tau_0K'(\alpha - 1)}
    & \leq \gamma/2,
\end{align}
where $w = W/m$.
Let $n$ be such that 
    \begin{align}
    \label{nRDPdef}
    1 - 2m(1 - w\tau_0)^n - 2mf(n,w\tau_0,\gamma')
    & \geq \conf
\end{align}
where 
$f$ is defined in~\cref{fdef}.
Finally, the returned value is
\begin{align}
    \label{Estdef}
    \frac{1}{\alpha-1} \log \sum_j 
        \left(
            \frac{N_j}{M_j}
        \right) ^\alpha \frac1n M_j.
\end{align}

\begin{theorem}[Correctness]
    \label{thm:RDP-correct}
    Let densities $p=p_{Z|X}(\cdot|x_1)$ and $q=p_{Z|X}(\cdot|x_2)$ over $\mathcal{Z} = [a,b]$ be $C$-Lipschitz, with $C<2/W^2$ and $W=b-a$. 
    For every $\alpha > 1$ (LRDP order), $\gamma > 0$ (precision), and $0<\conf<1$ (confidence), 
    it holds that
    \begin{align*}
       &\Pr\left(
       \begin{array}{l}
       \text{algorithm 3 succeeds and }
       \\
       |\epsilon_{\alpha}^{\star}(x_1, x_2) -\Tilde{\epsilon}_{\alpha}(x_1,x_2,{\cal Z},\gamma,\conf,C)|\leq\gamma 
       \end{array}
        \right )\geq \conf.
    \end{align*}
\end{theorem}

\subsection{Estimating $\epsilon_{\alpha}^{\star}(p_{Z|X})$}

We now focus on the estimation of \cref{eq:sup_rdp}. To this aim, we follow the approach in \cref{subsec:actEst}: we assume $D$-Lipschitzness of all $p_{Z|X}(z|\cdot)$, divide $\X$ in a proper number of buckets, run algorithm 3 for every pair of buckets' mid-points, and return the maximum returned value. The new algorithm, that we call Algorithm 4, is identical to \cref{alg:histEstimatorAll} except for choosing
\begin{align}
\label{kReniDef}
k \geq \frac {3(2\alpha-1)KD(d-c)} {2(\alpha-1)K'\tau_0\gamma}
\end{align}
and for invoking algorithm 3 in place of \cref{alg:histEstimator}.

\begin{theorem}
\label{thm:bucketsRDP}
Let $\Z=[a,b]$, $\X = [c,d]$, and $p_{Z|X}$ be such that, for every $x \in \cal X$, $p_{Z|X}(\cdot|x)$ is $C$-Lipschitz, for $C < 2/W^2$ and $W=b-a$, and that, for every $z \in \cal Z$, $p_{Z|X}(z|\cdot)$ is $D$-Lipschitz, for some $D$. For every
$\alpha > 1$ (LRDP order), $\gamma > 0$ (precision) and $0<\conf<1$ (confidence),
    we have that 
    \begin{align*}
       &\Pr\left(
       \begin{array}{l}
       \text{algorithm 4 succeeds and }
        \\
        |\epsilon_\alpha^{\star}(p_{Z|X}) -\Tilde{\epsilon}_\alpha({\cal Z},{\cal X},\gamma,\delta,C,D)|\leq\gamma
        \end{array}
        \right )\geq \conf.
    \end{align*}
\end{theorem}

\subsection{Experiments on LRDP}

We test the validity of our approach by repeating the same experiments as presented in the previous section about LDP but using the new algorithm for LRDP estimation. The results  confirm that our approach is practical and useful also in this framework. However, here it is particularly evident that the theoretical bound for sampling is much higher than the actual number needed in practice since now the theoretical bound is 10$^5$ times bigger than the practical one (whereas for LDP it was less than 10$^3$).

\bigskip
\noindent{\bf Number of samples. }
First of all, for LRDP the number of samples also depends on another parameter (apart from $\gamma$, $\conf$, and $C$), viz. $\alpha$. In \cref{fig:LogRDP} we depict the dependencies of $n$ as functions of the other parameters; as we can see, the higher dependency is on $\gamma$ and $C$, with a similar contribution.

\begin{figure}[t]
    \centering    \framebox{\includegraphics[width=.48\textwidth]{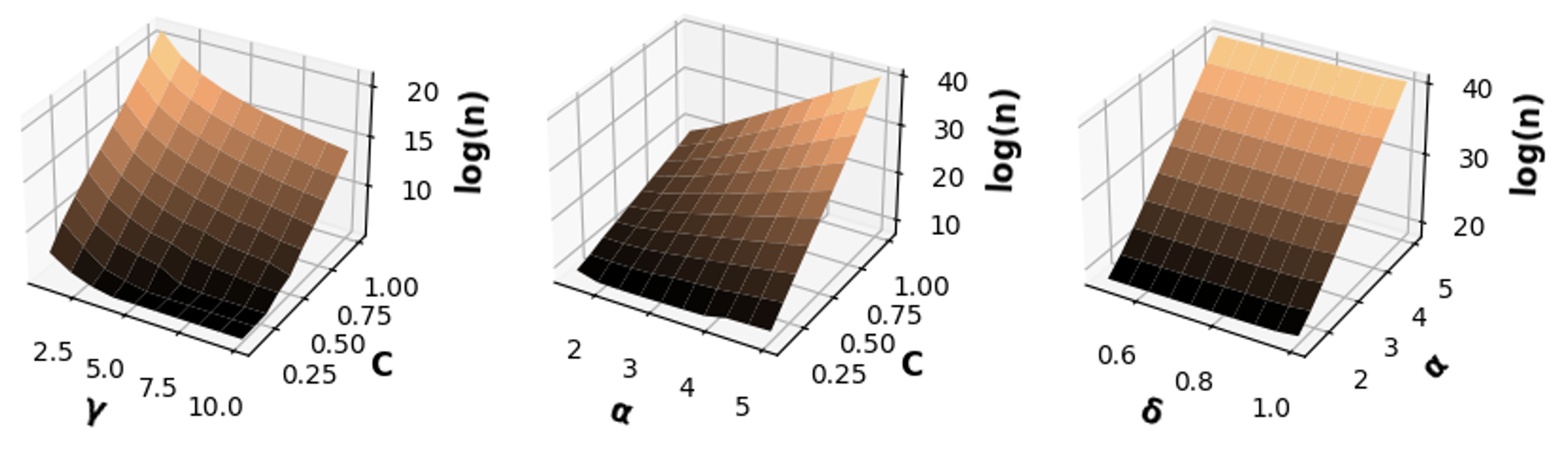}}
    \caption{Logarithm of the number of samples as functions of $\alpha, \conf, \gamma$ and $C$ for LRDP ($W=1$; when not varying, we set $\alpha = 2$, $C = 1$, $\conf = 0.9$ and $\gamma = 1$).}
    \label{fig:LogRDP}
\end{figure}


\begin{table}[h]
    \centering
    \resizebox{1\columnwidth}{!}{
    \begin{tabular}{c|c|c|||c|c|c}
    \toprule
    $x_i$ & $x_j$ &  $\Tilde{\epsilon}(x_i,x_j)$ &
    $x_i$ & $x_j$ &  $\Tilde{\epsilon}(x_i,x_j)$ \\
    \midrule
    0.00& 1.00& \textbf{0.027}& 0.55& 0.47& 0.000\\
    0.05& 0.87& 0.024& 0.61& 0.32& 0.005\\
    0.11& 0.71& 0.017& 0.66& 0.16& 0.013\\
    0.16& 0.55& 0.008& 0.71& 0.00& 0.018\\
    0.21& 0.39& 0.002& 0.79& 0.87& 0.000\\
    0.26& 0.24& 0.000& 0.84& 0.71& 0.001\\
    0.32& 0.08& 0.002& 0.89& 0.55& 0.006\\
    0.39& 0.95& 0.012& 0.95& 0.39& 0.012\\
    0.45& 0.79& 0.007& 1.00& 0.24& 0.021\\
    0.50& 0.63& 0.001& 1.00& 0.00& \textbf{0.027}\\    
    \bottomrule
    \end{tabular}
    }
    \caption{$\Tilde{\epsilon}(x_i,x_j)$ for some pairs 
    among the $k^2$ 
    computed by Algorithm 4 with a truncated Laplacian of parameter $B=3.5$, and with $\conf=0.9$ and $\gamma=0.5$.}
    \label{tab:est_all_rdp}
\end{table}

\begin{table*}[t]\centering
\resizebox{1.5\columnwidth}{!}{%
\begin{tabular}{r|ccc|ccc|ccc|ccc}
\toprule
\multirow{2}[3]{*}{$\conf=.9$} & \multicolumn{3}{c}
{\begin{tabular}{c}
     $B = 5$\\
     $\epsilon\approx .22$\\
     $C\approx .013$
\end{tabular}} & 
\multicolumn{3}{c}
{\begin{tabular}{c}
     $B = 3$\\
     $\epsilon\approx .39$\\
     $C\approx .036$
\end{tabular}} & 
\multicolumn{3}{c}
{\begin{tabular}{c}
     $B = 2$\\
     $\epsilon\approx .63$\\
     $C\approx .081$
\end{tabular}} & 
\multicolumn{3}{c}
{\begin{tabular}{c}
     $B = 1.5$\\
     $\epsilon\approx .91$\\
     $C\approx .142$
\end{tabular}} \\
\cmidrule(lr){2-4} \cmidrule(lr){5-7} \cmidrule(lr){8-10} \cmidrule(lr){11-13}
 & $n_{\text{TH}}$ & $n_{\text{PR}}$ & $m$ & $n_{\text{TH}}$ & $n_{\text{PR}}$ & $m$  & $n_{\text{TH}}$ & $n_{\text{PR}}$ & $m$  & $n_{\text{TH}}$ & $n_{\text{PR}}$ & $m$  \\
\midrule
$\gamma=1$ & 17794 & 56 & 3 & 2.3e5 & 56 & 10 & 6.7e6 & 293 & 41 & 3.4e8 & 1800 & 195 \\
$\gamma=.5$ & 1.6e5 & 56 & 6 & 2.1e6 & 150 & 20 & 5.7e7 & 650 & 81 & 3.0e9 & 3962 & 389 \\
$\gamma=.1$ & 2.5e7 & 831 & 29 & 3.1e8 & 2556 & 97 & 8.6e9 & 9812 & 403 & 4.3e11 & 45875 & 1945 \\
\bottomrule
\end{tabular}
}
\caption{Numbers of samples for different values of  $\gamma$ and $C$ for a truncated Laplace mechanism for LRDP. 
In every cell, we provide: the theoretical number of samples ($n_{\text{TH}}$), an order of magnitude for the practical one ($n_{\text{PR}}$), and the number of sub-intervals ($m$).}
\label{tab:RDP}
\end{table*}

\bigskip
\noindent{\bf Analysis of Algorithm 4. }
As we have done for LDP, we begin by analyzing the behavior of the estimator in Algorithm 4. We suppose the same experimental environment in~\cref{sec:exper} but with parameters $\gamma=0.5$, $\delta=0.9$, $B=3.5$ so $\epsilon\approx 0.027$, $D\approx 0.66$ and $C\approx 0.33$; moreover, we let $\alpha=2$ and, consequently, $k=39$.
In \cref{tab:est_all_rdp} we display $\Tilde{\epsilon}(x_i,x_j)$ for some pairs $x_i$ and $x_j$ among the $k^2(=1521)$ computed during an execution of Algorithm 4. Once again the algorithm correctly returns $\Tilde{\epsilon}=0.027$ and again this value is reached only for the pair of extremes $\{x_i,x_j\}=\{0,1\}$. Therefore, we will focus on $x_1=0$ and $x_2=1$ for the analysis of Algorithm~3.

\bigskip
\noindent{\bf Analysis of Algorithm 3. }
We show in~\cref{fig:RDPworks} the average value of the estimated $\epsilon$ across 100 executions of the proposed algorithm for LRDP.
As already anticipated, the most impressive result is the gap between the theoretical and the practical number of samples needed in order to satisfy~\cref{thm:RDP-correct}. Indeed, for LRDP the theoretical number of samples is way too large and the practical one is very small: there is a factor $\sim 8000$ between the two. To show this, we report in \cref{tab:RDP} a comparison between the two values for different sets of parameters.

\begin{figure}[t]
    \includegraphics[width=.45\textwidth]{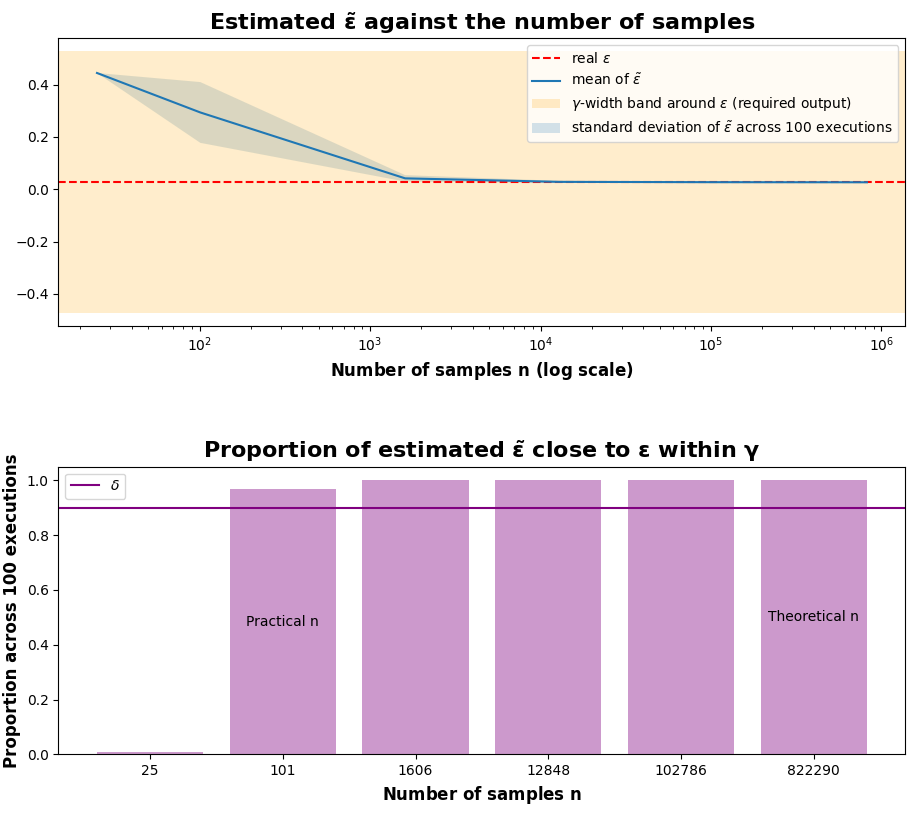}
    \caption{Analysis on $n$ across 100 execution of algorithm 2 for LRDP.}
    \label{fig:RDPworks}
\end{figure}

\section{Conclusion and Future work}
\label{sec:concl}

In this paper, we studied to what extent final users can infer information about the level of protection of their data in the differential privacy setting when the obfuscation mechanism is unknown to them (the so called ``black-box" scenario). We first proved that, without any assumption on the underlying distributions, it is not possible to estimate the degree of protection of the data, not even with low levels of precision and confidence. Our proof is for LDP and LRDP, but it can easily be adapted also to DP and RDP. By contrast, when the involved densities are Lipschitz in a closed interval, such guarantees exist for both $\varepsilon$-LDP and LRDP.
We validated our results by using one of the best known DP data obfuscation mechanisms (namely, the Laplace one).

This is just the starting point of our investigations. Possible directions for future research involve finding other possible assumptions on the distributions that allow for theoretical guarantees of the estimator, in particular for discrete distributions. In any case, as our impossibility results make evident, the crucial thing is to have distributions without peaks; so, some form of smoothness is needed.
Orthogonally, we can renounce to any theoretical guarantee and use estimators for ratios of probability distributions that may perform well in practice. Many works in this directions have been done (see, e.g., \cite{KHS08,KSS10,SSK12,LTSK17,K17}) and it would be interesting to see their performances in the setting of DP. This would also have the advantage of not making any assumption on the distributions involved.

\ifCLASSOPTIONcaptionsoff
  \newpage
\fi

\bibliographystyle{IEEEtran}
\bibliography{IEEEabrv,main.bib}

\newpage

\appendix
\subsection{Calculating a Truncated Laplace distribution}
\label{app_privacy:truncated}

For the experiments, we will need to massively sample from
the cumulative distribution function (CDF) of a
truncated Laplace distribution. However, classical libraries used for sampling (such as numpy for Python) do not implement such distributions.
Hence, we first determine the inverse CDF $F^{-1}$ of the truncated Laplace distribution $F$, and then apply it to values drawn from a uniform distribution on $[0,1]$ (implemented in classical libraries) to get values distributed according to the truncated Laplace. 
This is the well-known {\em inverse transform sampling} method used for pseudo-random number sampling (i.e., for generating sample numbers at random from any probability distribution) given its inverse CDF. The basic property of this method is that, if $U$ is a uniform random variable on $(0, 1)$, then $F^{-1}(U)$ has $F$ as its CDF. In our case, the required CDF is:\\

\vspace{-.4cm}
\noindent
$F(z|x,B)$
\begin{align*}
    & = \int_a^z f(u|x,B)du
    \\
    & =
    \begin{cases}
        BK_{x,B}
        \left(
            2 - e^{-(x - a)/B} - e^{-(z - x)/B}
        \right)
        \qquad \quad \text{if $z>x$}
        \\
        BK_{x,B}
        \left(
            e^{-(x - z)/B} - e^{-(x - a)/B}
        \right)
        \hfill \text{if $z \leq x$}
    \end{cases}
    \\
    & = BK_{x,B}
    \biggl(
        \operatorname{sgn}(z - x)
        \left(
            1 - e^{-|z - x|/B}
        \right) + 1 - e^{-(x - a)/B} 
    \biggr)
\end{align*}

To compute the inverse CDF, 
let $p\in [0,1]$. Then, we solve for $z$ the following equation:
\begin{align*}
    p = F(z|x,B)
\end{align*}
to set $F^{-1}(p|x,B) = z$ and we get\\

\noindent
$F^{-1}(p|x,B)$
\begin{align*}
    & =
    \begin{cases}
        x - B\log
        \left(
            2 - e^{-(x - a)/B} - \frac{p}{BK_{x,B}}
        \right) \qquad \quad \text{if $z > x$}
        \\
        x + B\log
        \left(
            e^{-(x - a)/B} + \frac{p}{BK_{x,B}}
        \right) \hfill \text{if $z \leq x$}
    \end{cases}
    \\
    & = x - \operatorname{sgn}(z - x)B\log
    \biggl[
        1\ +
    \\ & \qquad + \operatorname{sgn}(z - x)
        \biggl(
            1 - e^{-(x-a)/B} - \frac{p}{BK_{x,B}}
        \biggr)
    \biggr]
\end{align*}
Note that $F$ is strictly increasing and $F(z|x,B) = p$; so
\begin{align*}
    \operatorname{sgn}(z - x)
    & = \operatorname{sgn}(p-F(x|x,B))
\end{align*}

Therefore, we can derive the desired expression for $F^{-1}(p|x,B)$:
$$    
\begin{array}{l}
x - \operatorname{sgn}(p-F(x|x,B))B\log
    \biggl[
        1\ +
    \\ 
\quad + \operatorname{sgn}(p-F(x|x,B))
        \biggl(
            1 - e^{-(x-a)/B}
             - \frac{p}{BK_{x,B}}
        \biggr)
    \biggr].
\end{array}
$$
\subsection{Proofs}
\label{appendix:proofs}

\begin{IEEEproof}[Proof of \cref{lower}]
    Let $z\in\mathcal{Z}$. Then:
    \begin{align*}
        p(z) &= \left|p(z)-\frac{1}{W}+\frac{1}{W}\right|\\
        &\geq \frac{1}{W} - \left|p(z)-\frac{1}{W}\right| \tag{second triangle inequality}\\
        &=\frac{1}{W}-\left|\frac{1}{W}\int_a^b p(z)dt - \frac{1}{W}\int_a^b p(t)dt\right|\\
        &=\frac{1}{W}-\frac{1}{W}\left|\int_a^b (p(z)-p(t))dt\right|\\
        &\geq \frac{1}{W}-\frac{1}{W}\int_a^b |p(z)-p(t)|dt\\
        &\geq \frac{1}{W}-\frac{C}{W}\int_a^b |z-t|dt \tag{since $p$ is $C$-lipschitz}\\
        & = \frac{1}{W}-\frac{C}{W}\left[\int_a^z (z-t)dt + \int_z^b (t-z)dt\right]\\
        & = \frac{1}{W} - \frac{C}{W}\left[\frac{(z-a)^2}{2}+\frac{(b-z)^2}{2}\right]\\
        & \geq \frac{1}{W} - \frac{CW}{2}
    \end{align*}
    The last inequality comes from the fact that, if we consider $g(z)=(z-a)^2+(b-z)^2$,
    we have that
    $g'(z) = 4z - 2a - 2b \geq 0 \Longleftrightarrow z \geq \frac{a + b}{2}$.
    Hence, $g$ is decreasing and then increasing; therefore, its maximum is either $g(a)$ or $g(b)$. Since $g(a)=g(b)=W^2$, we have that $g(z)\leq W^2$.\\
    
    \noindent
    Similarly with the first triangle inequality we get
    
    \hfill \large $p(z) \leq \frac{1}{W}+\frac{CW}{2}$
\end{IEEEproof}

\bigskip

\begin{IEEEproof}[Proof of~\cref{thm:imposs}]
    Let $\rho=\sqrt[3]{1-\conf}\in(0,1)$; we shall
    write $\Tilde\epsilon$ instead of $\Tilde\epsilon(x_1,x_2,\gamma, \conf)$ and take $\mathcal{X}=\mathcal{Z}=\{0,1\}$. 
    For arbitrary $d \in (0,1)$ and $h > 1$, we set 
    $$p = p_{Z|X}(\cdot|0) = \mathcal{B}(d) \quad\text{and}\quad q = p_{Z|X}(\cdot|1) = \mathcal{B}\left(\frac{d}{h}\right)$$
    where $\mathcal{B}(\cdot)$ denotes a Bernoulli r.v.\,of parameter ‘$\cdot$'.
    Hence, $\Z(x_1,x_2)=\Z$ and
    \begin{align}
        \epsilon^\star(x_1,x_2) &=\sup_{x_1,x_2\in\mathcal{X}}\left[\sup_{z\in\mathcal{Z}(x_1,x_2)}\log\left(\frac{p_{Z|X}(z|x_1)}{p_{Z|X}(z|x_2)}\right)\right] \nonumber \\
        &=\max \left[\log\frac{1-d}{1-d/h},\log\frac{d}{d/h},\log\frac{1-d/h}{1-d},\log\frac{d/h}{d}\right] \nonumber \\
        &\geq \log(h) \label{lnh}
    \end{align}

    Now, let $N$ (still  depending on $x_1, x_2, \gamma,\conf$) be the (random) number of samples that $\Tilde\epsilon$ draws from either $p$ or $q$ during its execution. Let also $s_i$ denote the $i$-th sample from $p$ or $q$ returned by the sampler $\cal S$.
    For all $i$, we have that $\Pr_{s_i\sim p}(s_i=1)=d$ and $\Pr_{s_i\sim q}(s_i=1)=d/h$; so, in general
    \begin{align}
    \Pr(s_i=1)\leq d 
    \label{prSi}
    \end{align}
    Hence, if we let $P_n$ be the event \textit{“$\forall i \leq n. \ s_i = 0$"}, we can use \cref{prSi} and the union bound to have that:
    \begin{align}
        \Pr(P_n) 
        &= 1 - \Pr(\exists i \leq n. \ s_i = 1)\nonumber \\
        & \geq 1 - nd \label{1nd}
    \end{align}
    
    Then, for all $n$, consider the event $[N=n|P_n]$ (meaning that $\Tilde\epsilon$ stops sampling after $n$ draws equal to 0): it does not depend on the distribution $p_{Z|X}$ since $\Tilde\epsilon$ cannot have inferred any information on $p_{Z|X}$ with such samples. Hence, we can consider for a moment $p_{Z|X}=\mathcal{B}(0)$, the trivial obfuscation mechanism that always returns 0: $\Tilde\epsilon$ almost surely terminates also in this case and, in particular, it almost surely does a finite number $N_{\mathcal{B}(0)}$ of calls to the sampler. Hence, $\Pr(N_{\mathcal{B}(0)}=+\infty)=0$ and $\sum_{n\in\N}\Pr(N_{\mathcal{B}(0)}=n)=1$.
    Therefore, there exists $n_0$ such that
    \begin{align*}
        \sum_{n\leq n_0}\Pr(N_{\mathcal{B}(0)}=n)\geq \rho
    \end{align*}
    since $\rho<1$. Furthermore, for all $n$, it holds that
    $$\Pr(N_{\mathcal{B}(0)}=n)\ =\ \Pr(N=n|P_n)$$
    hence, 
        \begin{align}
        \sum_{n\leq n_0}\Pr(N=n|P_n)\geq \rho.
        \label{pralpha}
        \end{align}
    Now, for all $n\leq n_0$, let $h_n>1$ be such that         \begin{align}\label{hnprop}
    \Pr(\Tilde\epsilon+\gamma\leq \log(h_n)|P_n\,,N=n)\geq \rho.
    \end{align}
    Notice that such an $h_n$ exists for every $n$: indeed, since $\Tilde\epsilon$ has only sampled 0s, its answer does not depend on $h$.
    Now, set $h=\max_{n\leq n_0} h_n$ and $d=\frac{1-\rho}{n_0}$; then:
    
    $\Pr(|\epsilon^\star(x_1,x_2)-\Tilde\epsilon|\geq\gamma)$
    \begin{align*}
        & \geq \Pr(\Tilde\epsilon + \gamma \leq \epsilon^\star(x_1,x_2)) \tag{inclusion of events}\\
        & \geq \Pr(\Tilde\epsilon + \gamma \leq \log(h)) \tag{by~\eqref{lnh}}\\
        & = \sum_n \Pr(\Tilde\epsilon + \gamma \leq \log(h), N=n) \tag{law of total probability}\\
        & \geq \sum_n \Pr(\Tilde\epsilon + \gamma \leq \log(h), N=n, P_n)\\
        & \geq \sum_{n\leq n_0} \Pr(\Tilde\epsilon + \gamma \leq \log(h), N=n, P_n)\\
        & \geq \sum_{n\leq n_0} \Pr(\Tilde\epsilon + \gamma \leq \log(h_n), N=n, P_n) \tag{since $h\geq h_n$ for all $n\leq n_0$}\\
        & = \sum_{n\leq n_0} \Pr(\Tilde\epsilon + \gamma \leq \log(h_n) | N=n, P_n) \Pr(N=n,P_n) \tag{we set $\PR{A|B}=0$ whenever $\PR{B}=0$}\\
        & \geq \sum_{n\leq n_0} \rho \Pr(N=n,P_n) \tag{by \eqref{hnprop}}\\
        & = \sum_{n\leq n_0} \rho \Pr(N=n|P_n) \Pr(P_n)\\
        & \geq \sum_{n\leq n_0} \rho \Pr(N=n|P_n) (1-nd) \tag{by~\eqref{1nd}}\\
        & \geq \rho (1-n_0d) \sum_{n\leq n_0} \Pr(N=n|P_n) \tag{since $n\leq n_0$} \\
        & \geq \rho^2 (1-n_0d) \tag{by \eqref{pralpha}}\\
        & = \rho^3 \tag{by definition of $d$}\\
        & = 1 - \conf \tag{by definition of $\rho$}
    \end{align*}
\end{IEEEproof}

\bigskip

\begin{IEEEproof}[Proof of \cref{chernoff}] \ \\
    We have $\Pr(S_n>0) = 1-\Pr(S_n=0)=1-(1-p)^n$. Let us then define
    \begin{align}
        \Pr_0(A) &= \Pr(A|S_n>0)\nonumber
    \end{align}
    for any event $A$. Hence,
    \begin{align}
        \Pr_0(A) &= \Pr(A,S_n>0)/\Pr(S_n>0)\nonumber\\
        & \leq \Pr(A)/\Pr(S_n>0)\nonumber\\
        &= \frac{\Pr(A)}{1-(1-p)^n} \label{l21}
    \end{align}
    Then,
    \begin{align*}
        & \Pr_0(|\log S_n - \log np| > \gamma)
        \\
        &\quad = \Pr_0\text([\log S_n - \log np > \gamma] \text{ or } [\log np - \log S_n > \gamma])
        \\
        &\quad \leq \Pr_0(\log S_n - \log np > \gamma) + \Pr_0(\log np - \log S_n > \gamma)
        \\
        &\quad = \Pr_0(\log S_n > \log np+\gamma) + \Pr_0(\log S_n < \log np-\gamma)
        \\
        &\quad = \Pr_0(S_n > npe^\gamma) + \Pr_0(S_n < npe^{-\gamma})
        \\
        &\quad \leq \frac{\Pr(S_n > npe^\gamma) + \Pr(S_n < npe^{-\gamma})}{1-(1-p)^n}
        \tag{by~\eqref{l21}}
        \\
        &\quad = \big[\Pr(S_n > np(1+(e^\gamma - 1)))
        \\
        & \quad \quad \quad + \Pr(S_n < np(1 - (1-e^{-\gamma})))\big]/(1-(1-p)^n)
        \\
        &\quad = f(n,p,\gamma)
        \tag{multiplicative Chernoff bounds}
    \end{align*}
\end{IEEEproof}

\bigskip

To prove \cref{thm:estimator}, we need a preliminary lemma.

\begin{lemma}[Property of $\log$]
    \label[lemma]{logLipsch}
    For any $T > 0$ and $t,t' \geq T$, we have that
    $$
        |\log t - \log t'| \leq \frac{1}{T}|t - t'|
    $$
\end{lemma}
\begin{IEEEproof}
    Without loss of generality, we can suppose $t \leq t'$. Then the mean value theorem ensures that there exists $c \in [t, t']$ such that
    \begin{align*}
        \log t - \log t' = \log'(c) (t - t')
    \end{align*}
    Since \ $\log'(c) = \frac{1}{c} \leq \frac{1}{t} \leq \frac{1}{T}$, we also have that
    $$
        |\log t - \log t'|
        = \log'(c) |t - t'| 
        \leq \frac{1}{T} |t - t'|
        \vspace*{-.5cm}
    $$
\end{IEEEproof}

\medskip

\begin{IEEEproof}[Proof of~\cref{thm:estimator}] 
By \eqref{mdef}, we have that
    \begin{align}
        \frac{Cw}{\tau} &\leq \frac{\gamma}{6} \label{m}
    \end{align}
    Now, let us define
    \begin{align*}
        p_j &\defsym \Pr(z_j \leq {\cal S}(x_1) < z_{j+1}) = \int_{z_j}^{z_{j+1}} p(t)dt \\
        q_j &\defsym \Pr(z_j \leq {\cal S}(x_2) < z_{j+1}) = \int_{z_j}^{z_{j+1}} q(t)dt
    \end{align*}
    the probabilities of sampling in $[z_j,z_{j+1})$ from $p$ and $q$, respectively. We first prove that (the log of) $\frac{p_j}{q_j}$ is close to (the log of) $\frac{p(z)}{q(z)}$.
    \medskip
    \begin{quote}
        \label{cl1}
        \textbf{Claim 1:} {\it for all $j$ and $z \in [z_j,z_{j+1})$, it holds that} $$\left|\log\left(\frac{p_{j}}{q_{j}}\right)-\log\left(\frac{p(z)}{q(z)}\right)\right|\leq \frac{\gamma}{6}.$$
    \end{quote}
    First, thanks to~\cref{lower} and \eqref{mdef}, we have that
    \begin{align}
        &p(z)\geq \tau \text{ and } q(z)\geq \tau \text{ for all $z$.} 
        \label{pqtau}
    \end{align}
    So,
    \begin{align}
    p_j = \int_{z_j}^{z_{j+1}} p(t)dt \geq \int_{z_j}^{z_{j+1}}\tau dt = w\tau \label{pqj}
    \end{align}
    and similarly $q_j \geq w\tau$.
    Then, for all $j$ and $z \in [z_j , z_{j+1})$:
    \begin{align*}
    & \left|\log\left(\frac{p_{j}}{q_{j}}\right)-\log\left(\frac{p(z)}{q(z)}\right)\right|
    \\
    & \qquad = \left|\log\left(\frac{p_{j}}{q_{j}}\right)-\log\left(\frac{wp(z)}{wq(z)}\right)\right|\\
        &\qquad \leq |\log p_{j} - \log wp(z)|+|\log q_{j} - \log wq(z)|
        \\
        & \qquad \leq \frac{1}{w\tau}\biggl[|p_{j}-wp(z)|+|q_{j}-wq(z)|\biggr]
        \tag{by \cref{logLipsch}, \eqref{pqtau} and \eqref{pqj}}
    \end{align*}
    Now,
    \begin{align}
        |p_{j}-wp(z)| &=
        \left|
            \int_{z_{j}}^{z_{j+1}} (p(t)-p(z))dt
        \right| \nonumber\\
        &\leq \int_{z_{j}}^{z_{j+1}} |p(t)-p(z)|dt \nonumber\\
        &\leq C\int_{z_{j}}^{z_{j+1}} |t-z|dt \tag{$p$ is $C$-lipschitz}\\
        &\leq \frac{Cw^2}{2}
        \label{boundpj}
    \end{align}
    with the last inequality resulting from the same reasoning as for the end of the proof of~\cref{lower}.
    The same applies for $q$, and so by~\eqref{m}
    \begin{align}
        \left|\log\left(\frac{p_{j}}{q_{j}}\right)-\log\left(\frac{p(z)}{q(z)}\right)\right|\nonumber
        & \leq \frac{1}{w\tau}\biggl[\frac{Cw^2}{2}+\frac{Cw^2}{2}\biggr]
        \leq \frac{\gamma}{6}
    \end{align}
    This proves \hyperref[cl1]{Claim 1}.
    \vspace{.6cm}
    
    Since $\epsilon^\star(x_1,x_2) = \sup_{z\in\mathcal{Z}}\log\left(\frac{p(z)}{q(z)}\right)$ and $\mathcal{Z}=[a,b]$ is closed, $\epsilon^\star(x_1,x_2)$ is reached for some $z_0\in\mathcal{Z}$. Let $j_0$ be such that $z_{j_0} \leq z_0 < z_{j_0+1}$ ($z_0$ is in the sub-interval of index $j_0$).
    Let also the random variable $J=\text{argmax}_{1\leq j\leq m}\log\left(\frac{N_j}{M_j}\right)$ be the index of the return of our estimator in case of success.
    
    For all $j$, let 
    $$C_{j,\gamma} \defsym  ``\ \left|\log\left(\frac{N_j}{M_j}\right)-\log\left(\frac{p_j}{q_j}\right)\right| \leq \frac{\gamma}{6}\ "
    $$
    denote the event that we estimate precisely enough the ratio for the sub-interval of index $j$. 
    Furthermore, let 
    $$E \defsym ``\text{success}\text{ and } C_{J,\gamma} \text{ and } C_{j_0,\gamma}"$$
    denote the event that our estimator succeeds (so for all $j$, $N_j>0$ and $M_j>0$) and it precisely estimates the ratio for the sub-intervals of index $J$ and $j_0$.
    We now show that $E$ implies that the whole estimation process succeeds.
    
    \bigskip
    \begin{quote}
    \label{cl2}
        \textbf{Claim 2:} {\it If $E$ then $\left|\log\left(\frac{N_J}{M_J}\right)-\epsilon^\star(x_1,x_2)\right|\leq \gamma$.}
    \end{quote}
    Indeed, suppose that we have run our algorithm and that $E$ occurs; let us then show that the answer of our estimator is close to $\epsilon^\star(x_1,x_2)$. To this aim, 
    let us prove that our answer (viz., $\log \frac{N_J}{M_J}$) is close to $\log \frac{p_{j_0}}{q_{j_0}}$, since (thanks to \hyperref[cl1]{Claim 1}) we know that the latter is close to $\epsilon^\star(x_1,x_2) = \log \frac{p(z_0)}{q(z_0)}$. 
    By contradiction, suppose that:
    \begin{align}
    \left|\log\left(\frac{N_J}{M_J}\right)-\log\left(\frac{p_{j_0}}{q_{j_0}}\right)\right|>\frac{\gamma}{2}
    \label{contrHyp}
    \end{align}
    Hence, $J\neq j_0$ and
    \eqref{contrHyp} can hold because of two possible cases, which we now examine in isolation. In what follows, we shall use the observation, referred to as $(\dagger)$, that $|x-y|\leq z$ implies $x-y\leq z$, for every $x,y,z$ (indeed, $x-y\leq |x-y|$).
    \begin{itemize}
    \item Case 1:
    $\log\left(\frac{N_J}{M_J}\right)>\log\left(\frac{p_{j_0}}{q_{j_0}}\right)+\frac{\gamma}{2}$.
    Then, for any $z \in [z_J, z_{J+1})$, we have that
    \begin{align*}
        \log\left(\frac{p(z)}{q(z)}\right)
        & \geq \log\left(\frac{p_J}{q_J}\right) - \frac{\gamma}{6} \tag{by \hyperref[cl1]{Claim 1} and $(\dagger)$}\\
        & \geq \log\left(\frac{N_J}{M_J}\right) - \frac{2\gamma}{6} \tag{by $(\dagger)$ and since $C_{J,\gamma}$ occurs}\\
        & > \log\left(\frac{p_{j_0}}{q_{j_0}}\right) + \frac{\gamma}{2}- \frac{\gamma}{3}\\
        & \geq \log\left(\frac{p(z_0)}{q(z_0)}\right) - \frac{\gamma}{6}+\frac{\gamma}{6} \tag{by \hyperref[cl1]{Claim 1} and $(\dagger)$}\\
        & = \epsilon^\star(x_1,x_2)
    \end{align*}
    which is impossible by definition of $\epsilon^\star(x_1,x_2)$.
    
    \item Case 2:
    $\log\left(\frac{p_{j_0}}{q_{j_0}}\right)>\log\left(\frac{N_J}{M_J}\right)+\frac{\gamma}{2}$. Then:
    \begin{align*}
        \log\left(\frac{N_{j_0}}{M_{j_0}}\right) & \geq \log\left(\frac{p_{j_0}}{q_{j_0}}\right) - \frac{\gamma}{6} \tag{by $(\dagger)$ and since $C_{j_0,\gamma}$ occurs}\\
        & > \log\left(\frac{N_J}{M_J}\right)+\frac{\gamma}{2} - \frac{\gamma}{6}\\
        & \geq \log\left(\frac{N_J}{M_J}\right)
    \end{align*}
    which is impossible by definition of $J$ (the algorithm should have answered $\log \frac{N_{j_0}}{M_{j_0}}$).
    \end{itemize}
    
    \noindent
    Therefore, 
    \begin{align*}
    \left|\log\left(\frac{N_J}{M_J}\right)-\epsilon^\star(x_1,x_2)\right| & 
    \leq \left|\log\left(\frac{N_J}{M_J}\right)-\log\left(\frac{p_{j_0}}{q_{j_0}}\right)\right|
    \\
    & \quad + \left|\log\left(\frac{p_{j_0}}{q_{j_0}}\right)-\log\left(\frac{p(z_0)}{q(z_0)}\right)\right|
    \\
    & \leq \gamma
    \end{align*}
    because of the negation of \eqref{contrHyp}
    and by \hyperref[cl1]{Claim 1}. This proves \hyperref[cl2]{Claim 2}.\\

    Finally, for showing that our estimator succeeds with high enough probability, we will prove that $E$ occurs with high enough probability. To this aim, first observe that
    \begin{align*}
        \Pr(\text{fail}) &= \Pr(\exists j\ s.t. \ N_j = 0 \text{ or } M_j = 0)\\&\leq \sum_j [\Pr(N_j = 0) + \Pr(M_j = 0)]
    \end{align*}
    However,
    \begin{align*}
        \Pr(N_j = 0) &= \Pr(\forall i. \ s_1[i] \notin [z_j, z_{j+1}))\\
        &= \Pr({\cal S}(x_1) \notin [z_j,z_{j+1}))^n \tag{the $s_1[i]$'s are i.i.d.}\\
        &= (1-\Pr({\cal S}(x_1) \in [z_j,z_{j+1})))^n\\
        &= (1-p_j)^n \tag{definition of $p_j$}\\
        &\leq (1-w\tau)^n \tag{by~\eqref{pqj}}
    \end{align*}
    and similarly, $\Pr(M_j=0)\leq (1-w\tau)^n$.
    Then
    \begin{align}
        \Pr(\text{fail}) \leq 2m(1-w\tau)^n \label{fail}
    \end{align}
    
    Furthermore, we have that all
    $N_j$ (resp., $M_j$)
    are binomial distributions of parameters $n$ and $p_j$ (resp. $q_j$); so
    \begin{align}
    & \PR{C_{j,\gamma}^c \ \bigg| \ N_j>0 \text{ and } M_j > 0} 
    \nonumber\\
    & \quad = \PR{\left|\log\left(\frac{N_j}{M_j}\right)-\log\left(\frac{np_j}{nq_j}\right)\right| > \frac{\gamma}{6} \ \bigg| \ N_j, M_j > 0}\nonumber\\
        & \quad\leq \PR{\left|\log N_j-\log np_j\right| > \frac{\gamma}{12} \ \bigg| \ N_j>0}\nonumber\\
        & \quad\qquad + \PR{\left|\log M_j-\log nq_j\right| > \frac{\gamma}{12} \ \bigg| \ M_j>0}\nonumber \tag{$N_j$ and $M_j$ are independent}\\
        & \quad\leq f(n,p_j,\frac \gamma{12})+f(n,q_j,\frac\gamma{12}) \tag{by~\cref{chernoff}}\nonumber\\
        & \quad\leq 2f(n,w\tau,\frac\gamma{12}) \label{csucc}
    \end{align}
    
    To conclude, we have that:
    \begin{align*}
    & \PR{\text{success}, \left|\log\left(\frac{N_J}{M_J}\right)-\epsilon^\star(x_1,x_2)\right|\leq \gamma}
    \\
    & \quad \geq \PR{E,\left|\log\left(\frac{N_J}{M_J}\right)-\epsilon^\star(x_1,x_2)\right|\leq \gamma}\\
        & \quad= \PR{E} \tag{by \hyperref[cl2]{Claim 2}}\\
        & \quad= 1 - \PR{E^c}\\
        & \quad= 1 - \PR{\text{fail} \text{ or } [C_{J,\gamma}^c,\text{success}] \text{ or } [C_{j_0,\gamma}^c,\text{success}]}
        \tag{since $A^c\cup B^c=A^c\cup (B^c\setminus A^c)=A^c\cup(B^c\cap A$)}\\
        & \quad\geq 1 - \PR{\text{fail}} - \PR{C_{J,\gamma}^c,\text{success}} - \PR{C_{j_0,\gamma}^c,\text{success}}\\
        & \quad\geq 1 - \PR{\text{fail}}\\
        & \quad\qquad - \PR{C_{J,\gamma}^c,N_J>0,M_J>0}\\
        & \quad\qquad - \PR{C_{j_0,\gamma}^c,N_{j_0}>0,M_{j_0}>0}\\
        & \quad\geq 1 - \PR{\text{fail}}\\
        & \quad\qquad - \PR{C_{J,\gamma}^c \ \bigg| \ N_J>0,M_J>0}\\
        & \quad\qquad - \PR{C_{j_0,\gamma}^c, \ \bigg| \ N_{j_0}>0,M_{j_0}>0}
        \tag{$\PR{A,B}=\PR{A|B}\PR{B}\leq \PR{A|B}$}\\
        & \quad\geq 1-2m(1-w\tau)^n-4f(n,w\tau,\gamma/12) \tag{by~\eqref{fail} and~\eqref{csucc}}\\
        & \quad\geq \conf \tag{by \eqref{ndef}}
    \end{align*}
    This proves the theorem.
\end{IEEEproof}

\bigskip

\begin{IEEEproof}[Proof of \cref{thm:scLip}] By using the same notation as the proof of \cref{thm:estimator}, for all $j$, we have that
    \begin{align*}
        \left|\frac{1}{n}N_j-\frac{1}{n}N_{j+1}\right|
        \leq & \left|\frac{1}{n}N_j-p_j\right|
        + |p_j-p_{j+1}|\\
        & + \left|p_{j+1}-\frac{1}{n}N_{j+1}\right|
    \end{align*}
    First, notice that
    \begin{align*}
        |p_j-p_{j+1}|&\leq |p_j-wp(z_{j+1})|+|wp(z_{j+1})-p_{j+1}|\\
        & \leq Cw^2 \tag{by \eqref{boundpj}}
    \end{align*}
    Then, for any $c\geq0$, notice that the event
    $$\left|\frac{1}{n}N_j-p_j\right|\leq c \ \text{ and } \ \left|\frac{1}{n}N_{j+1}-p_{j+1}\right|\leq c$$
    implies the event $\left|\frac{1}{n}N_j-\frac{1}{n}N_{j+1}\right|\leq 2c+Cw^2$. Therefore,\\
    $\PR{\left|\frac{1}{n}N_j-\frac{1}{n}N_{j+1}\right|\leq 2c+Cw^2}$
    \begin{align*}
        &\geq \PR{\left|\frac{1}{n}N_j-p_j\right|\leq c \ \text{ and } \ \left|\frac{1}{n}N_{j+1}-p_{j+1}\right|\leq c}\\
        &=1-\PR{\left|\frac{1}{n}N_j-p_j\right|> c \ \text{ or } \ \left|\frac{1}{n}N_{j+1}-p_{j+1}\right|> c}\\
        &\geq 1 - \PR{\left|\frac{1}{n}N_j-p_j\right|> c} - \PR{ \left|\frac{1}{n}N_{j+1}-p_{j+1}\right|> c}
    \end{align*}
    Furthermore, for every $j$, we have that
    \begin{align*}
        \PR{\left|\frac{1}{n}N_j-p_j\right|> c}
        &=\PR{\left|N_j-np_j\right|> np_j\frac{c}{p_j}}\\
        &\leq 2\exp\left(\frac{-np_j(c/p_j)^2}{3}\right)\tag{multiplicative Chernoff bound}
        \\
        &\leq 2\exp\left(\frac{-nc^2}{3}\right)
    \end{align*}
    Hence,
    \begin{align*}
        \PR{\left|\frac{1}{n}N_j-\frac{1}{n}N_{j+1}\right|\leq 2c+Cw^2}\geq 1-4\exp\left(\frac{-nc^2}{3}\right)
    \end{align*}
    Since this is true also for $M_j$, we can conclude that
    \begin{align*}
        \Pr\biggl(\exists j. \ &
        \left|\frac{1}{n}N_j-\frac{1}{n}N_{j+1}\right|> 2c+Cw^2 \text{ or }\\
        & \left|\frac{1}{n}M_j-\frac{1}{n}M_{j+1}\right|> 2c+Cw^2\biggr)
    \end{align*}
    \qquad \qquad $\leq 8m\exp\left(\frac{-nc^2}{3}\right)$
\end{IEEEproof}

\bigskip

\begin{IEEEproof}[Proof of \cref{thm:buckets}]
Since both $\X$ and $\Z$ are closed, $\epsilon^{\star}(p_{Z|X})$ is reached on some $x_1^\star, x_2^\star \in \X$ and $z^\star \in \Z$. Let $i^\star$ and $j^\star$ be the buckets where $x_1^\star$ and $x_2^\star$ fall into. Then, 
\begin{align}
\label{lab:boundk}
|x_1^\star - x_{i^\star}| \leq \frac{d-c}{2k} \quad\text{and}\quad 
|x_2^\star - x_{j^\star}| \leq \frac{d-c}{2k}.
\end{align}
Let us now call 
\begin{align*}
L \defsym \log{\frac{p_{Z|X}(z^\star|x_{i^\star})}{p_{Z|X}(z^\star|x_{j^\star})}}
\end{align*}
Since $\epsilon^\star(p_{Z|X}) = \log{\frac{p_{Z|X}(z^\star|x_1^\star)}{p_{Z|X}(z^\star|x_2^\star)}}$, we have that
\begin{align*}
& |\epsilon^\star(p_{Z|X}) - L|
\\
& \quad \leq |\log{p_{Z|X}(z^\star|x_1^\star)} - \log{p_{Z|X}(z^\star|x_{i^\star})}| 
\\
& \qquad + |\log{p_{Z|X}(z^\star|x_2^\star)} - \log{p_{Z|X}(z^\star|x_{j^\star})}|
\\
& \quad \leq \frac 1 \tau |p_{Z|X}(z^\star|x_1^\star) - p_{Z|X}(z^\star|x_{i^\star})| 
\\
& \qquad + \frac 1 \tau |p_{Z|X}(z^\star|x_2^\star) - p_{Z|X}(z^\star|x_{j^\star})|
\tag{by \cref{lower} and \cref{logLipsch}}
\\
& \quad \leq \frac D \tau |x_1^\star - x_{i^\star}| 
+ \frac D \tau |x_2^\star - x_{j^\star}|
\tag{by $D$-Lipschitzness}
\\
& \quad \leq \frac{D(d-c)}{\tau k}
\tag{by \cref{lab:boundk}}
\\
& \quad \leq \frac \gamma 3
\tag{by \cref{kdef00} of \cref{alg:histEstimatorAll}}
\end{align*}
Since by definition $\epsilon^\star(x,x') = \sup_{z \in \Z}\log{\frac{p_{Z|X}(z|x)}{p_{Z|X}(z|x')}}$
and $\epsilon^\star(p_{Z|X}) = \sup_{x,x' \in \X}\epsilon^\star(x,x')$, we have that
$\epsilon^\star(p_{Z|X}) \geq \epsilon^\star(x_{i^\star},x_{j^\star}) \geq L$.
Thus, by the above inequality, we also have that
\begin{align}
\label{epstarDiff}
\epsilon^\star(p_{Z|X}) - \epsilon^\star(x_{i^\star},x_{j^\star})
\leq \frac \gamma 3.
\end{align}
To streamline writing, let us write 
$\Tilde{\epsilon}(x_i,x_i,{\cal Z},\frac \gamma 3,\sqrt\delta,C) \defsym \Tilde\epsilon_{ij}$ and $\epsilon^\star(x_i,x_j)\defsym \epsilon^\star_{ij}$.
Now assume that \cref{alg:histEstimator} succeeds in estimating $\epsilon^\star_{i^\star j^\star}$ with precision $\gamma/3$; then, \cref{alg:histEstimatorAll} succeeds, in the sense specified before \cref{thm:buckets}. Let us now consider the bucket indexes $I$ and $J$ of the answer of the algorithm i.e.
$$\Tilde{\epsilon}({\cal Z},{\cal X},\gamma,\delta,C,D) =
\Tilde\epsilon_{IJ}.$$
By \cref{thm:estimator}, we have that \cref{alg:histEstimator} succeeds with
\begin{align*}
    |\epsilon^\star_{i^\star j^\star} - \Tilde\epsilon_{i^\star j^\star}| \leq \frac\gamma 3
    \ \text{ and } \ 
    |\epsilon^\star_{IJ} - \Tilde\epsilon_{IJ}|& \leq \frac\gamma 3
\end{align*}
with probability $\sqrt\delta$.
By construction of the algorithm, we know that $\Tilde\epsilon_{i^\star j^\star} \leq \Tilde\epsilon_{IJ}$.
If by contradiction $\Tilde\epsilon_{IJ} > \epsilon^\star_{i^\star j^\star} + 2\gamma/3$, then
\begin{align*}
    \epsilon^\star_{IJ} & \geq \Tilde\epsilon_{IJ} - \gamma/3\\
    & > \epsilon^\star_{i^\star j^\star} + \gamma/3\\
    & \geq \epsilon^\star(p_{Z|X}) \tag{by \cref{epstarDiff}}
\end{align*}
which is absurd. So, $\Tilde\epsilon_{IJ} \leq \epsilon^\star_{i^\star j^\star} + 2\gamma/3$ and as 
$$\Tilde\epsilon_{IJ} \geq \Tilde\epsilon_{i^\star j^\star} \geq \epsilon^\star_{i^\star j^\star} - \frac \gamma 3$$
then
$$
    |\Tilde\epsilon_{IJ} - \epsilon^\star_{i^\star j^\star}| \leq \frac{2\gamma}{3}.
$$
Thus,
\begin{align}
    |\epsilon^\star(p_{Z|X}) - \Tilde\epsilon_{IJ}| & \leq |\epsilon^\star(p_{Z|X}) - \epsilon^\star_{i^\star j^\star}| + |\epsilon^\star_{i^\star j^\star} - \Tilde\epsilon_{IJ}| \notag\\
    & \leq \frac \gamma 3 + \frac{2\gamma}{3} = \gamma \label{alg2works}
\end{align}
meaning that \cref{alg:histEstimatorAll} estimates $\epsilon^\star(p_{Z|X})$ with the required precision $\gamma$. Let us show that this happens with probability at least $\delta$:
\begin{align*}
& \Pr(\text{Alg.\ref{alg:histEstimatorAll} succeeds and } |\epsilon^{\star}(p_{Z|X}) -\Tilde{\epsilon}|\leq\gamma)
\\
& \quad = \Pr(\text{Alg.\ref{alg:histEstimatorAll} succeeds and } |\epsilon^{\star}(p_{Z|X}) - \epsilon^\star(x_{i^\star},x_{j^\star})
\\
& \hspace*{4.3cm}+\epsilon^\star(x_{i^\star},x_{j^\star}) - \Tilde{\epsilon}|\leq\gamma)
\\
& \quad \geq \Pr\left(\text{Alg.\ref{alg:histEstimatorAll} succeeds and } | \epsilon^\star(x_{i^\star},x_{j^\star}) - \Tilde{\epsilon}|\leq \frac{2\gamma}{3}\right)
\tag{
by \cref{epstarDiff}}
\\
& \quad \geq \Pr\left(\text{Alg.\ref{alg:histEstimator} succeeds for $(x_I,x_J)$ and $(x_{i^\star}, x_{j^\star})$}\right)
\tag{as the latter event implies the former one}
\\
& \quad = (\sqrt{\delta})^2 = \delta \tag{by independence and \cref{thm:estimator}}
\end{align*}
\end{IEEEproof}

\bigskip

\begin{IEEEproof}[Proof of \cref{thm:RenyiImposs}]
The proof of \cref{thm:imposs} can be adapted to RDP with
$$
    p = \mathcal{B}
    \left(
        d^{1/\alpha}
    \right)
    \ \text{ and } \
    q = \mathcal{B}
    \left(
        \left[
            \frac{d}{h^{\alpha - 1}}
        \right]^{1/(\alpha - 1)}
    \right)
$$
Indeed, analogously to \eqref{lnh}, we have that
    \begin{align*}
        \epsilon_\alpha^\star(x_1,x_2)
        & \geq D_\alpha(p\|q)
        \\
        & = \frac{1}{\alpha-1} \log
            \E_{z \sim q}
            \left(
                \frac{p(z)}{q(z)}
            \right)^\alpha
        \\
        & = \frac{1}{\alpha-1} \log
            \left[
                \left(
                    \frac{p(0)}{q(0)}
                \right)^\alpha
                q(0)
                +
                \left(
                    \frac{p(1)}{q(1)}
                \right)^\alpha
                q(1)
            \right]
        \\
        & \geq \frac{1}{\alpha-1} \log 
        \left[
            \left(
                \frac{p(1)}{q(1)}
            \right)^\alpha q(1)
        \right]
        \tag{as log is increasing}
        \\
        & = \log(h)
    \end{align*}
    Moreover, we have that
    $$
        \Pr_{s_i \sim p}(s_i = 1) = d^{1/\alpha}
        \leq d^{1/(\alpha - 1)}
    $$
    and
    $$
        \Pr_{s_i \sim q}(s_i = 1)
        =
        \left[
            \frac{d}{h^{(\alpha-1)}}
        \right]^{1/(\alpha-1)}
        \leq d^{1/(\alpha - 1)}
    $$
    So, similarly to \eqref{prSi}, we have $\Pr(s_i = 1) \leq d^{1/(\alpha - 1)}$ and, similarly to \eqref{1nd}, we have
    \begin{align*}
        \PR{P_n} & \geq 1 - nd^{1/(\alpha-1)}
    \end{align*}
    If we choose $d = \left(\frac{1-\rho}{n_0}\right)^{\alpha-1}$, we can conclude like in the proof of \cref{thm:imposs}, since $1-n_0d^{1/(\alpha-1)} = \rho$.
\end{IEEEproof}

\bigskip

To prove \cref{thm:RDP-correct}, we need a preliminary lemma, analogous to \cref{logLipsch}.

\begin{lemma}[Property of $\exp$]
    \label[lemma]{expLipsch}
    For any $T > 0$ and $t,t' \leq T$, we have that
    \begin{align*}
        |t - t'| \leq T|\log t - \log t'|
    \end{align*}
\end{lemma}
\begin{IEEEproof}
    Without loss of generality we can suppose $t \leq t'$. Let $s = \log t$ and $s' = \log t'$; hence, $s \leq s'$. The mean value theorem ensures that there exists $c \in [s,s']$ such that
    \begin{align*}
        e^s - e^{s'} = \exp'(c)(s - s')
    \end{align*}
    and, since \ $\exp' c = e^c \leq e^{s'} = t' \leq T$, we have that
    $$
        |t - t'|
        = |e^s - e^{s'}|
        = \exp'(c)|s - s'| \leq T|\log t - \log t'|
        \vspace*{-.5cm}
    $$
\end{IEEEproof}

\medskip

\begin{IEEEproof}[Proof of \cref{thm:RDP-correct}]
Let us define
    \begin{align*}
        p_j & \defsym \Pr(z_j \leq {\cal S}(x_1) < z_{j+1}) = \int_{z_j}^{z_{j+1}} p(t)dt \\
        q_j & \defsym \Pr(z_j \leq {\cal S}(x_2) < z_{j+1}) = \int_{z_j}^{z_{j+1}} q(t)dt
    \end{align*}
    the probabilities of sampling in $[z_j,z_{j+1})$ from $p$ and $q$, respectively.
    Thanks to \cref{lower}, for all $t$ we have that
    \begin{align}
        \tau_0 \leq p(t) \leq \tau_1
        \ \text{ and } \
        \tau_0 \leq q(t) \leq \tau_1
        \label{tau01}
    \end{align}
    and so
    \begin{align}
        w\tau_0 \leq p_j \leq w\tau_1
        \ \text{ and } \
        w\tau_0 \leq q_j \leq w\tau_1.
        \label{wtau01}
    \end{align}
    For all $j$, let $D_{j,\gamma'}$ denote the event
    $$
    |\log N_j - \log np_j| \leq \gamma'
    \ \text{ and } \
    |\log M_j - \log nq_j| \leq \gamma'
    $$ 
    in the case of success, so that it is well defined.
The key result for proving the theorem is the following claim, where we denote with $Est$ the quantity in \cref{Estdef}.

    \medskip
    \begin{quote}
    \label{cl3}
    \textbf{Claim 3:} {\it 
    Suppose that algorithm 3 succeeds and all the $D_{j,\gamma'}$ occur. Then, 
    $\left|D_\alpha(p\|q) - Est\right| \leq \gamma$.}
    \end{quote}
\noindent
    If algorithm 3 succeeds and all the $D_{j,\gamma'}$ occur, then
    $$
    np_je^{-\gamma'} \leq N_j \leq np_je^{\gamma'}
    \ \text{ and } \
    nq_je^{-\gamma'} \leq M_j \leq nq_je^{\gamma'}
    $$
    for all $j$; therefore,
    \begin{align}
        \frac{N_j^\alpha}{nM_j^{\alpha-1}}
        & \leq \frac{(np_je^{\gamma'})^\alpha}
                    {n(nq_je^{-\gamma'})^{\alpha-1}}
        \nonumber
        \\
        & \leq \frac{(nw\tau_1e^{\gamma'})^\alpha}
        {n(nw\tau_0e^{-\gamma'})^{\alpha-1}}
        \tag{by \eqref{wtau01}}
        \nonumber
        \\
        & \leq w\frac{\tau_1^\alpha e^{\gamma'(2\alpha-1)}}
                    {\tau_0^{\alpha-1}}
        \nonumber
        \\
        & \leq wK
        \label{wk2}
    \end{align}
    as $e^{\gamma'(2\alpha - 1)} \leq 2$ (indeed, by
    \eqref{defGammaPr}, $\gamma' \leq \frac{\log 2}{2\alpha-1}$).
    Similarly,
    \begin{align}
        \frac{N_j^\alpha}{nM_j^{\alpha-1}}
        \geq wK'
        \label{wk'}
    \end{align}
    
    Moreover, because of \cref{tau01}, we have that
    \begin{align}
        wK'
        \leq w \frac{p(t)^\alpha}{q(t)^{\alpha - 1}}
        \leq wK
        \label{wk1}
    \end{align}

    Now, thanks to $C$-lipschitzness, \eqref{boundpj} holds also in the setting of this proof.
    Hence, for all $j$ and for all $z \in [z_j, z_{j + 1})$, we can apply \cref{logLipsch} with $T = w\tau_0$, $t = p_j$ and $t' = wp(z)$ (as, by \eqref{wtau01}, $p_j \geq w\tau_0$ and $wp(z) \geq w\tau_0$) that, together with \eqref{boundpj}, gives
    \begin{align}
        |\log p_j - \log wp(z)|
        & \leq \frac{1}{w\tau_0} |p_j - wp(z)|
        \leq \frac{Cw}{2\tau_0}
        \label{cw2}
    \end{align}
    and the same holds for $q$.
    
    By the mean value theorem, we have that, for all $j$, there exists some $t_j \in (z_j, z_{j + 1})$ such that
    \begin{align}
        \int_{z_j}^{z_{j+1}}
            \frac{p(t)^\alpha}{q(t)^{\alpha-1}}
            dt & 
            = w\frac{p(t_j)^\alpha}{q(t_j)^{\alpha-1}}
        \label{meanValThm}
    \end{align}
    Therefore, we have:\\
    
    \noindent
    $\left|\int_{z_j}^{z_{j+1}}\frac{p(t)^\alpha}{q(t)^{\alpha-1}}dt-\frac{N_j^\alpha}{nM_j^{\alpha-1}}\right|$
    \begin{align}
        & = \left|
                w\frac{p(t_j)^\alpha}{q(t_j)^{\alpha-1}}
                - \frac{N_j^\alpha}{nM_j^{\alpha-1}}
            \right|
            \tag{by \eqref{meanValThm}}
        \\
        & \leq wK 
                \left|
                    \log w\frac{p(t_j)^\alpha}
                              {q(t_j)^{\alpha-1}}
                    - \log \frac{N_j^\alpha}
                                {nM_j^{\alpha-1}}
                \right|
        \tag{by \cref{expLipsch} with \eqref{wk1} and \eqref{wk2}}
        \\
        & = wK 
                \left|
                    \log \frac{(wp(t_j))^\alpha}
                              {(wq(t_j))^{\alpha-1}}
                    -
                    \log \frac{p_j^\alpha}
                              {q_j^{\alpha-1}}
                \right.
                \nonumber
                \\
        & \qquad\qquad\qquad\qquad \left.
                    + \log \frac{p_j^\alpha}
                                {q_j^{\alpha-1}}
                    - \log \frac{N_j^\alpha}
                                {nM_j^{\alpha-1}}
                \right|
                \nonumber
        \\
        & = wK 
                \left|
                    \log \frac{(wp(t_j))^\alpha}
                              {(wq(t_j))^{\alpha-1}}
                    -
                    \log \frac{p_j^\alpha}
                              {q_j^{\alpha-1}}
                    \right.
                    \nonumber
                    \\
        & \qquad\qquad\qquad\qquad
                \left.
                    + \log \frac{(np_j)^\alpha}
                                {(nq_j)^{\alpha-1}}
                    - \log \frac{N_j^\alpha}
                                {M_j^{\alpha-1}}
                \right|
                \nonumber
        \\
        \nonumber
        \\
        & \leq wK 
        \biggl(
            \alpha|\log wp(t_j) - \log p_j|
            + (\alpha-1)|\log wq(t_j) - \log q_j|
            \nonumber
            \\
            & \quad \qquad + \alpha|\log np_j - \log N_j|
            + (\alpha-1)|\log nq_j - \log M_j|
        \biggr)
        \nonumber
        \\
        & \leq wK(2\alpha-1)
                \left(
                    \gamma' + \frac{Cw}{2\tau_0}
                \right)
                \label{boundTwo}
    \end{align}
    where the last step holds by \eqref{cw2} and since $D_{j,\gamma'}$ occurs.
    
    Thanks to \eqref{meanValThm} and \eqref{wk1} (for \eqref{WKK})  and to \eqref{wk'} (for \eqref{WK'}), we have
    \begin{align}
        & \E_{t\sim q}
        \left(
            \frac{p(t)}{q(t)}
        \right)^\alpha
        = \int_a^b 
            \left(
                \frac{p(t)}{q(t)}
            \right)^\alpha q(t) dt
        \geq WK' 
        \label{WKK}
        \\
        & \sum_j 
            \left(
                \frac{N_j}{M_j}
            \right)^\alpha \frac1n M_j
        \geq mwK' = WK'
        \label{WK'}
    \end{align}
    Therefore:\\

    \noindent
    $\left|D_\alpha(p\|q) - Est\right|$
    \begin{align}
        & = \frac{1}{\alpha - 1}
            \left|
                \log \int_a^b 
                \left(
                    \frac{p(t)}{q(t)}
                \right)^\alpha q(t) dt
                - \log \sum_j 
                \left(
                    \frac{N_j}{M_j}
                \right)^\alpha \frac1n M_j
            \right|
            \nonumber
        \\
        & \leq \frac{1}{(\alpha - 1)WK'} \left|
                \int_a^b 
                    \left(
                        \frac{p(t)}{q(t)}
                    \right)^\alpha q(t) dt
                - \sum_j 
                    \left(
                        \frac{N_j}{M_j}
                    \right)^\alpha \frac1n M_j
            \right|
        \tag{by \cref{logLipsch} and \eqref{WK'}}
        \\
        & \leq \frac{1}{(\alpha - 1)WK'} \sum_j
                \left|
                    \int_{z_j}^{z_{j+1}} 
                    \left(
                        \frac{p(t)}{q(t)}
                    \right)^\alpha q(t) dt
                    - 
                    \left(
                        \frac{N_j}{M_j}
                    \right)^\alpha \frac1{n} M_j
                \right|
                \nonumber
        \\
        & \leq \frac{mwK(2\alpha-1)}{(\alpha - 1)WK'}
            \left(
                \gamma'+\frac{Cw}{2\tau_0}
            \right)
            \tag{by \eqref{boundTwo}}
        \\
        & = \frac{K(2\alpha-1)}{K'(\alpha-1)}
            \left(
                \gamma'+\frac{Cw}{2\tau_0}
            \right)
            \nonumber
        \\
        & \leq \gamma
        \nonumber
    \end{align}
    The last step holds since $\gamma'$ and $\frac{Cw}{2\tau_0}$ are at most $\frac{\gamma K'(\alpha-1)}{2K(2\alpha-1)}$, because of \eqref{defGammaPr} and \eqref{defM}. 
This completes the proof of \hyperref[cl3]{Claim 3}.

\medskip

With this result, the correctness proof can be completed:
\vspace*{.15cm}

    \noindent
    $\Pr\left(
    \begin{array}{l}
    \text{algorithm 3 succeeds and }
    \\
    |\epsilon^\star_\alpha(x_1,x_2)-\Tilde\epsilon_\alpha(x_1,x_2,\mathcal{Z},\gamma,\conf,C)|\leq\gamma
    \end{array}
    \right)
    \vspace*{-.15cm}
    $
    \begin{align*}
        & = \PR{\text{success and } \left|D_\alpha(p\|q) - Est\right| \leq \gamma}
        \\
        & \geq \PR{\text{success} \text{ and } \forall j \ D_{j,\gamma'}}
        \tag{by \hyperref[cl3]{Claim 3}}
        \\
        & = 1 - \PR{
                    \text{fail or }
                    \exists j \
                    (D_{j,\gamma'}^c \text{ and success})
                }
        \\
        & \geq 1 - \PR{\text{fail}}
        - \PR{\exists j \ (D_{j,\gamma'}^c \text{ and success})}
        \\
        & \geq 1 - \PR{\exists j \ (N_j = 0 \ \text{ or } \ M_j = 0)}
        \\
        & \qquad - \sum_j \PR{D_{j,\gamma'}^c \text{ and success}}
        \\
        & \geq 1 - \sum_j
            \big[
                \PR{N_j = 0} + \PR{M_j = 0}
            \big]
        \\
        & \qquad - \sum_j \PR{
            |\log N_j - \log np_j| > \gamma' 
            \text{ and success}}
        \\
        & \qquad - \sum_j \PR{
            |\log M_j - \log nq_j| > \gamma' 
            \text{ and success}}
        \\
        & \geq 1 - \sum_j ((1 - p_j)^n + (1 - q_j)^n)
        \\
        & \qquad - \sum_j \PR{
            |\log N_j - \log np_j| > \gamma' 
            \ \big| \ N_j > 0}
        \\
        & \qquad - \sum_j \PR{
            |\log M_j - \log nq_j| > \gamma' 
            \ \big| \ M_j > 0}
        \\
        & \geq 1 - \sum_j 2(1 - w\tau_0)^n
            - \sum_j f(n, p_j, \gamma')
            - \sum_j f(n, q_j, \gamma')
        \tag{by \eqref{wtau01} and \cref{chernoff}}
        \\
        & \geq 1 - 2m(1 - w\tau_0)^n
            - 2mf(n,w\tau_0,\gamma')
        \\
        & \geq \conf
        \tag{by definition of $n$}
    \end{align*}
\end{IEEEproof}

\bigskip
\bigskip

\begin{IEEEproof}[Proof of \cref{thm:bucketsRDP}]
Let $x_1^\star, x_2^\star \in \X$ be the values where
$\epsilon_\alpha^{\star}(p_{Z|X})$ is reached, and $i^\star$ and $j^\star$ be the buckets where $x_1^\star$ and $x_2^\star$ fall into. 
To streamline writing, we abbreviate
$p_{Z|X}(\cdot|x_1^\star)$ and $p_{Z|X}(\cdot|x_2^\star)$
as $p_{1}(\cdot)$ and $p_{2}(\cdot)$; similarly,
we write $p_{i^\star}(\cdot)$ and $p_{j^\star}(\cdot)$
for $p_{Z|X}(\cdot|x_{i^\star})$ and $p_{Z|X}(\cdot|x_{j^\star})$.

Since $\epsilon_\alpha^\star(p_{Z|X}) = D_\alpha(p_1(\cdot)\big\| p_2(\cdot))$,
$\epsilon_\alpha^\star(x_{i^\star},x_{j^\star}) = D_\alpha(p_{i^\star}(\cdot) \big\| p_{j^\star}(\cdot))$, 
and 
$D_\alpha(p_{Z|X}(\cdot|x),p_{Z|X}(\cdot|x')) = \frac 1 {\alpha-1}\log{\int_a^b\frac{p_{Z|X}(z|x)^\alpha}{p_{Z|X}(z|x')^{\alpha-1}}dz}$, we have that
\begin{align*}
& |\epsilon_\alpha^\star(p_{Z|X}) - \epsilon_\alpha^\star(x_{i^\star},x_{j^\star})| =
\\
& = \frac 1 {\alpha-1} \left|
\log{\int_a^b\frac{p_1(z)^\alpha}{p_2(z)^{\alpha-1}}dz}
-
\log{\int_a^b\frac{p_{i^\star}(z)^\alpha}{p_{j^\star}(z)^{\alpha-1}}dz}
\right|
\\
& \quad \leq \frac 1 {(\alpha-1)WK'} \left|
\int_a^b\frac{p_1(z)^\alpha}{p_2(z)^{\alpha-1}}dz
-
\int_a^b\frac{p_{i^\star}(z)^\alpha}{p_{j^\star}(z)^{\alpha-1}}dz
\right|
\tag{by \cref{logLipsch} and \cref{WKK}}
\\
& \quad \leq \frac 1 {(\alpha-1)WK'} \sum_{t=0}^{m-1} \int_{z_t}^{z_{t+1}}\left|
\frac{p_1(z)^\alpha}{p_2(z)^{\alpha-1}}
-
\frac{p_{i^\star}(z)^\alpha}{p_{j^\star}(z)^{\alpha-1}}
\right|dz
\\
& \quad = \frac w {(\alpha-1)WK'} \sum_{t=0}^{m-1} \left|
\frac{p_1(z'_t)^\alpha}{p_2(z'_t)^{\alpha-1}}
-
\frac{p_{i^\star}(z'_t)^\alpha}{p_{j^\star}(z'_t)^{\alpha-1}}
\right|
\tag{by the mean value thm., with $z'_t \in (z_t,z_{t+1})$}
\\
& \quad \leq \frac {wK} {(\alpha-1)WK'} \sum_{t=0}^{m-1} \left|
\log\frac{p_1(z'_t)^\alpha}{p_2(z'_t)^{\alpha-1}}
-
\log\frac{p_{i^\star}(z'_t)^\alpha}{p_{j^\star}(z'_t)^{\alpha-1}}
\right|
\tag{by \cref{expLipsch} and \cref{wk1}}
\\
& \quad \leq \frac {wK} {(\alpha-1)WK'} \sum_{t=0}^{m-1} (
\alpha|\log{p_1(z'_t)}
- \log{p_{i^\star}(z'_t)}|
\\
& \hspace*{3.5cm} + (\alpha-1)|\log{p_2(z'_t)}
-\log{p_{j^\star}(z'_t)}|
)
\\
& \quad \leq \frac {wK} {(\alpha-1)WK'} \sum_{t=0}^{m-1} \left(
\frac \alpha {\tau_0}|p_1(z'_t)
- p_{i^\star}(z'_t)|
\right.
\\
& \hspace*{3.7cm} + \left. \frac {\alpha-1}{\tau_0} |p_2(z'_t)
- p_{j^\star}(z'_t)|
\right)
\tag{by \cref{logLipsch} and \cref{tau01}}
\\
& \quad \leq \frac {wK} {(\alpha-1)WK'} \sum_{t=0}^{m-1}\!\! \left(
\frac {\alpha D |x_1^\star-x_{i^\star}|} {\tau_0}
 + \frac {(\alpha-1)D|x_2^\star - x_{j^\star}|}{\tau_0}
\right)
\tag{by $D$-Lipschitzness}
\\
& \quad \leq 
\frac {wKm(2\alpha-1)D(d-c)} {(\alpha-1)WK'\tau_02k}
\tag{by \cref{lab:boundk}}
\\
& \quad =
\frac {(2\alpha-1)KD(d-c)} {2(\alpha-1)K'\tau_0k}
\tag{by def. of $w$}
\\
& \quad \leq \frac \gamma 3
\tag{by def. of $k$ for Alg.4 (\cref{kReniDef})}
\end{align*}
The proof then proceeds like the one for \cref{thm:buckets}.
\end{IEEEproof}

\end{document}